\documentclass[draftclsnofoot,onecolumn,11pt]{IEEEtran} 
 
\newcommand{\size}{N}

\newcommand{\radius}{\rho} 
\newcommand{\bsize}{6w^3}
\newcommand{\bfsize}{2w^3}
\newcommand\given[1][]{\:#1\vert\:}

\usepackage{mathrsfs} 

\newcommand\ddfrac[2]{\frac{\displaystyle #1}{\displaystyle #2}}

\newcommand{\qed}{\nobreak \ifvmode \relax \else 
      \ifdim\lastskip<1.5em \hskip-\lastskip
      \hskip1.5em plus0em minus0.5em \fi \nobreak
      \vrule height0.75em width0.5em depth0.25em\fi} 
      
\usepackage[backend=bibtex,style=ieee,sorting=none]{biblatex}
\bibliography{Refs}

\usepackage{graphicx}
\graphicspath{{EPS/}}

\usepackage[cmex10]{amsmath}
\usepackage{amssymb}
\usepackage{bm}
\usepackage{commath}

\usepackage{amsthm}

\usepackage{enumitem}

\usepackage{array}

\newtheorem{theorem*}{\textbf{Theorem}}
\newtheorem{proposition*}{\textbf{Proposition}}

\newtheorem{lemma*}{\textbf{Lemma}}
\newtheorem{conjecture*}{\textbf{Conjecture}}

\newtheorem{corollary*}{Corollary}

\usepackage[hidelinks]{hyperref}

\begin{document}

\title{
Self-organized Segregation on the Grid
\thanks{This is a pre-print of an article published in Journal of Statistical Physics. The final authenticated version is available online at: \url{https://doi.org/10.1007/s10955-017-1942-4}. This work was partially supported by Army Research Office (ARO), award number W911NF-15-1-0253.   An extended abstract of   this paper  has appeared in the proceedings of ACM Symposium on Principles of Distributed Computing 2017 without rigorous proofs and with some results omitted.}
 }



 \author{Hamed Omidvar, 
Massimo Franceschetti \\
Department of Electrical and Computer Engineering, \\University of California, San Diego \\
\texttt{ \{homidvar, mfrances\}@ucsd.edu}}

\maketitle

\maketitle

\begin{abstract}
We consider an agent-based model with exponentially distributed waiting times  in which two types of agents interact locally over a graph,  and based on this interaction and on the value of a common intolerance threshold $\tau$, decide whether to change their types. This is equivalent to  a zero-temperature Ising model with Glauber dynamics, an Asynchronous Cellular Automaton (ACA) with extended Moore neighborhoods, or a Schelling model of self-organized segregation  in an open system, and has applications in the analysis of social and biological networks, and spin glasses systems. Some rigorous results were recently obtained in the theoretical computer science literature, and this work provides several extensions. We enlarge the intolerance interval leading to the expected formation of  large segregated regions  of agents of a single type from the known size $\epsilon>0$ to size $\approx 0.134$. Namely, we show that for $0.433 < \tau < 1/2$ (and by symmetry $1/2<\tau<0.567$), the  expected size of the largest segregated region containing  an arbitrary agent  is exponential in the size of the neighborhood. 
We further extend the interval  leading to  expected large segregated regions to size $\approx 0.312$ considering ``almost segregated'' regions, namely regions where the ratio of the  number of agents of one type and the  number of agents of the other type vanishes quickly as the size of the neighborhood grows. In this case,  we show that for $0.344 < \tau \leq 0.433$   (and by symmetry for $0.567 \leq \tau<0.656$) the expected size of the  largest almost segregated region containing an arbitrary agent is exponential in the size of the neighborhood. This behavior is reminiscent of supercritical percolation, where small clusters of empty sites can be observed within any sufficiently large region of the occupied  percolation cluster. The exponential bounds that we provide  also imply that complete segregation, where agents of a single type cover the whole grid, does not occur with high probability for $p=1/2$ and the range of intolerance considered. 
\end{abstract}

\section{Introduction}
\label{intro}
\subsection{Background}
A basic observation   made by Thomas Schelling while   studying the mechanisms leading to social segregation in the   United States \cite{schelling1969models,schelling1971dynamic} was that   individuals in a social network   have interactions  with their friends and neighbors rather than with the entire population,  and this often triggers  global effects that were not originally intended, nor desired.   Schelling proposed a simple 
 stochastic model to predict these global outcomes, which has become  popular in the social sciences. Two types of agents are randomly placed at the vertices of a two-dimensional grid and interact with a small subset of nodes located in their local neighborhood. Based on these interactions, the boolean state of each agent is determined as follows. All agents have a common intolerance threshold,   indicating the minimum fraction  of agents of their same type that must be located in their neighborhood to make their state happy.  Unhappy agents   randomly move to vacant locations where they will be happy. A peculiar effect observed by simulating several variants of this model is that when the system reaches a stable state, large areas of segregated agents of the same type are observed, for a wide range of the intolerance threshold value. Individuals, Schelling concluded, tend to spontaneously self-segregate. See Figure~\ref{Fig:Sim_results_1} for a simulation of this behavior.
 
\begin{figure}
\begin{center}
{\includegraphics[width=\textwidth]{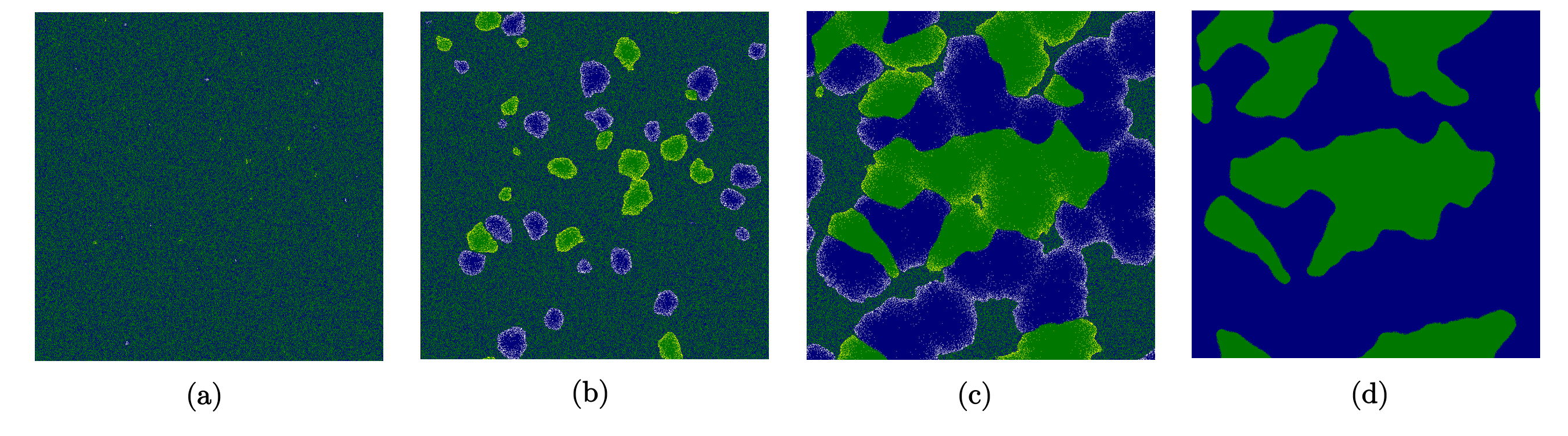}}
\end{center}
\caption{{\footnotesize Self-segregation arising over time for a value of the intolerance $\tau=0.42$ on a grid of size $1000 \times 1000$ and neighborhood size $441$. Green and blue indicate areas of ``happy'' agents of type (+1) and (-1), respectively. White and yellow indicate areas of ``unhappy" agents of type (+1) and (-1) respectively.  Initial configuration (a),   intermediate configurations (b)-(c),  final configuration (d). When the process terminates all agents are happy but large segregated regions can be observed. }}
\label{Fig:Sim_results_1}
\end{figure}

%


Similar models have been considered in the statistical physics literature  well before Schelling's observation. For an intolerance value of $1/2$, for example,  agents take the same value of the majority of their neighbors, and self-organized segregation in the Schelling model corresponds to spontaneous magnetization in the Ising model with zero temperature, where spins align along the direction of the local field~\cite{stauffer2007ising, castellano2009}. In computation theory, mathematics, physics, complexity theory, theoretical biology and microstructure modeling, the model is known as  a two-dimensional, two-state Asynchronous Cellular Automaton (ACA) with extended Moore neighborhoods and exponential waiting times \cite{chopard1998cellular}.  Other related models appeared in epidemiology \cite{hethcote2000mathematics, draief2010epidemics}, economics \cite{jackson2002formation}, engineering and computer sciences \cite{kleinberg2007cascading, easley2010networks}. 
Mathematically, all  of these models fall in the general area of interacting particle systems, or contact processes, and exhibit phase transitions~\cite{liggett2012interacting,liggett2013stochastic}. 

Schelling-type models can be roughly divided into two classes.
A   \textit{Kawasaki dynamic} model assumes
  there are  no vacant positions in the underlying graph, and a pair of unhappy agents swap their locations  if this will make both of them happy.  A \textit{Glauber dynamic} model assumes single agents to simply flip their types if this makes them happy. This flipping action   indicates that the agent has moved out of the system and a new agent has occupied its location. While in a Kawasaki model   the  system is ``closed" and the number of agents of the same type is fixed, in a Glauber model the system is ``open" and the number of agents of the same type may change over time. Sometimes the model dynamics are defined as having unhappy agents swap (or flip) regardless of whether this makes them happy or not. We  assume throughout Glauber dynamics and  agents to  flip only if this makes them happy.
Another possible variant is to assume that agents have a small probability of acting differently than what the general rule prescribes, 
other variants also consider having multiple intolerance levels,  multiple agent types, different agent distributions, and time-varying intolerance  \cite{young2001individual, zhang2004dynamic, zhang2004residential, zhang2011tipping, mobius2000formation, meyer2003immigration,bhakta2014clustering, schulze2005potts,barmpalias2015minority,barmpalias2015tipping}.


\subsection{Contribution}
We focus on  the  case of two types of agents placed  uniformly  at random on a two-dimensional grid according to a Bernoulli distribution of parameter $p  = 1/2$ and having  a single intolerance level $0<\tau<1$, and study the range of intolerance leading to the formation of large segregated regions.  Even for the one-dimensional version of this problem  rigorous results appeared only recently.   

Brandt et al.~\cite{brandt2012analysis}  considered  a  ring graph for the Kawasaki model of evolution.  In this setting, letting the  neighborhood of an agent be the set of nearby agents that is used to determine whether the agent is happy or not, they  showed that    for an intolerance level $\tau=1/2$,  the expected  size of the largest segregated region containing  an arbitrary agent in steady state is polynomial in the size of the neighborhood.  Barmpalias et al.~\cite{barmpalias2014digital} showed
that there exists a value of $\tau^*\approx 0.35$, such that for all  $\tau<\tau^*$ the initial configuration remains static with high probability (w.h.p.), while for all $\tau^*<\tau<1/2$  
the   size of the largest segregated region in steady state  becomes   exponential in the size of the neighborhood w.h.p. On the other hand, for all $\tau>1/2$  the system evolves w.h.p.\ towards a  state with only two segregated components. For the   Glauber model the behavior is similar, but symmetric around $\tau=1/2$, with a first transition from a static configuration to exponential segregation occurring at $\tau \approx 0.35$, a special point $\tau=1/2$ with the largest segregated  region of expected polynomial size,  then again exponential segregation until $\tau \approx 0.65$, and finally a static configuration for larger values of $\tau$.


In a two-dimensional grid graph on a  torus, the case $\tau=1/2$ is open. Immorlica et al.~\cite{immorlica2015exponential}  have shown for the Glauber model  the existence of a value $\tau^*< 1/2$, such that for all $\tau^*<\tau<1/2$  the  expected size of the largest segregated region is exponential in the size of the neighborhood. This shows that segregation is expected  in the small interval $\tau \in (1/2-\epsilon, 1/2)$.  
Note that this does not imply exponential segregation w.h.p., but only expected segregated regions of exponential size.  Barmpalias et al.~\cite{barmpalias2016unperturbed} considered a model in which each type of agent has a different intolerance, i.e.,  $\tau_1$ and $\tau_2$. For the special case of $\tau_1 = \tau_2 = \tau$, they have shown that when  $\tau>3/4$, or  $\tau<1/4$,  the initial configuration remains static  w.h.p.
 
Our main contribution is depicted in Figure~\ref{fig:tau}.  We consider the Glauber model for the two-dimensional grid graph on a torus. First,  we enlarge the intolerance interval that leads to the formation of large segregated regions from the known size $\epsilon>0$ to size $\approx 0.134$, namely  we show that when $0.433 < \tau < 1/2$ (and by symmetry $1/2<\tau<0.567$), the  expected size of the largest segregated region is exponential in the size of the neighborhood. 
Second, we further extend the interval  leading to large segregated regions to size $\approx 0.312$. In this case,  the main contribution is that we  consider 
``almost segregated''  regions, namely regions where the ratio of the  number of agents of one type and the  number of agents of the other type quickly vanishes   as the size of the neighborhood grows,
 and show that for $0.344 < \tau \leq 0.433$   (and by symmetry for $0.567 \leq \tau<0.656$) 
the expected size of the largest  almost segregated region  is   exponential in the size of the neighborhood.

As  shown for the one dimensional case in~\cite{barmpalias2014digital} and   conjectured for the two-dimensional case in~\cite{barmpalias2016unperturbed}, we show that as the intolerance parameter gets farther from one half, in both directions, the average size of both the segregated and almost segregated regions gets larger: higher tolerance in our model does not necessarily lead to less segregation. On the contrary,  it can increase the size of the segregated areas. This result is depicted in Figure~\ref{fig:aaprime}. The intuitive explanation is that  highly tolerant agents are seldom unhappy in the initial configuration, and the segregated regions of  opposite types that unhappy agents may  ignite are likely to start from far apart, and may grow larger  before meeting at their boundaries. 

Finally, the exponential upper bound that we provide on the expected size of the largest segregated region implies that complete segregation, where agents of a single type cover the whole grid, does not occur w.h.p. for the range of intolerance considered.  In contrast, Fontes et al.~ \cite{siv} have shown the existence of a critical probability  $1/2<p^*<1$ for the initial Bernoulli distribution of the agents such that for $\tau =1/2$  and $p > p^*$ the Glauber model on the $d$-dimensional grid converges to a state where only one type of agents are present.
This shows that complete segregation occurs w.h.p.\ for $\tau=1/2$ and $p \in (1-\epsilon,1)$. Morris~\cite{morris2011zero} has shown that $p^*$   converges to $1/2$ as $d \rightarrow \infty$. Caputo and Martinelli \cite{caputo2006phase} have shown the same result for $d$-regular trees, while  Kanoria and Montanari~\cite{montanaritree} derived it for  $d$-regular trees in a synchronous setting where  flips occur simultaneously,  and obtained lower bounds on $p^*(d)$ for small values of $d$. The case $d = 1$   was first investigated by Erd\"{o}s and Ney~\cite{erdos1974some}, and Arratia~\cite{arratia1983site} has proven that $p^*(1)=1$. 

\subsection{Techniques}
Our proofs are based on a typicality argument showing a self-similar structure of the neighborhoods in the initial state of the process, and on the identification of geometric configurations igniting a cascading process leading to segregation.
 We make extensive use of     tools from percolation theory, including the exponential decay of the radius of the open  cluster below criticality~\cite{grimmett1999percolation},
concentration bounds on the passage time~\cite{kesten1993speed} (see also \cite{damron2014subdiffusive, talagrand1995concentration}), and on the chemical distance between percolation sites~\cite{garet2007large}. We also make frequent use of  renormalization, and correlation inequalities for contact processes~\cite{liggett2010stochastic}.  In this framework, we provide an extension of the Fortuin-Kasteleyn-Ginibre (FKG) inequality in a dynamical setting that can be of independent interest.
\begin{figure}[!t] 
\centering
\includegraphics[width=3in]{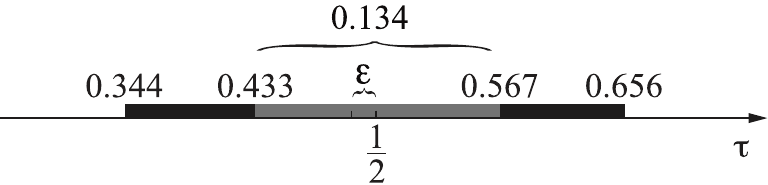}
\caption{We enlarge the width of the intolerance interval for which the expected size of the largest segregated region containing an arbitrary agent is exponential in the size of the neighborhood from the known value $\epsilon>0$ to $\approx 0.134$ (grey region). We also show that the expected size of the largest almost segregated region containing an arbitrary agent is exponential in the size of the neighborhood for an intolerance interval of width $\approx 0.312$ (grey plus black region).  }
\label{fig:tau}
\end{figure}

The  paper is organized as follows. In section \ref{Sec:Model} we introduce the model, state our results, and give a summary of the proof construction. In section~\ref{Sec:Triggering} we study the initial configuration and derive some properties of the sub-neighborhoods of the unhappy agents. 
In section \ref{Sec:Segregation} we study the dynamics of the segregation process and derive the main results. Concluding remarks are given in section \ref{Sec:Conclusion}.

\section{Model and main results} \label{Sec:Model}

\subsection{The Model} 
\textit{Initial Configuration.}
We consider an $n \times n$ grid graph $G_n$ embedded on a torus $\mathbb{T}=[0,n)\times[0,n)$, an integer $w \in O(\sqrt{\log n})$ called \emph{horizon}, and a rational $0 \leq \tau \leq 1$  called \emph{intolerance}.  All arithmetic operations over the coordinates are performed modulo $n$, i.e., $(x,y)=(x+n,y)=(x,y+n)$. We place an agent at each node of the grid and choose its type independently at random  to be (+1) or (-1) according to a Bernoulli distribution of parameter $p  = 1/2$.

\

A \textit{neighborhood} is a connected sub-graph of $G_n$. 
A  \textit{neighborhood of radius ${\radius}$} is the set of all agents with $l_\infty$ distance  at most ${\radius}$ from a central node, and is denoted by  $\mathcal{N}_{\radius}$. The \textit{size} of a neighborhood is the  number of agents in  it.   
The \textit{neighborhood of an agent $u$} is a neighborhood of radius equal to the horizon and centered at $u$, and is denoted by $\mathcal{N}(u)$.

\

\textit{Dynamics.}
 We let the rational $\tau$ called \emph{intolerance} be $\lceil \tilde{\tau} {\size} \rceil/{\size}$, where $\tilde{\tau} \in [0,1]$   and 
${\size} = (2w+1)^2$  is the size of the neighborhood of an  agent. The integer $\tau {\size}$ represents the minimum number of agents of the same type as $u$ that must be present in $\mathcal{N}(u)$  to make   $u$  \emph{happy}. More precisely,
for every agent $u$, we let $s(u)$ be the ratio between the number of agents of the same type as $u$ in its neighborhood and the size of the neighborhood. 
At any point in continuous-time, if ${s(u) \ge \tau}$ then $u$ is labeled \textit{happy}, otherwise it is labeled \textit{unhappy}. We assign independent and identical Poisson clocks to all agents, and every time a clock rings the type of the agent is flipped if and only if the agent is unhappy and this flip will make the agent happy.
 Two observations are now  in order. First, for $\tau< 1/2$  flipping its type will always make an unhappy agent happy,  but this is not the case  for $\tau>1/2$. Second, the process dynamics are equivalent to a discrete-time model where at each discrete time step one unhappy agent is chosen uniformly at random and its type is flipped if this will make the agent happy.


\

\textit{Termination.}
The process continues until there are no unhappy agents left, or there are no unhappy agents that can become happy by flipping their type. By defining a Lyapunov function to be  the sum over all agents $u$ of the number of  agents of the same type as $u$ present in its neighborhood, it is easy to argue that the process   indeed terminates.

\

\textit{Segregation.}
The \textit{monochromatic region} of an agent $u$ is the neighborhood with largest radius containing agents of a single type and  that  also  contains $u$ when the process stops. 
Let    $\epsilon>0$ and ${\size}=(2 w+1)^2$.
The \textit{almost monochromatic region} of an agent $u$, is the neighborhood with largest radius such that the ratio of the number of agents of one type and the number of agents of the other type is bounded by $e^{-{\size}^{\epsilon}}$ and that also contains $u$ when the process stops.

\

Throughout the paper we use the terminology \textit{with high probability} (w.h.p.) meaning that the probability of an event approaches one as $N$ approaches infinity. 

\subsection{The Results}
To state our results, 
we let  $\tau_1 \approx 0.433$ be the solution of 
\begin{align}
\frac{3}{4} \left[1-H\left(\frac{4}{3} \tau_1 \right)\right]- \left[1-H\left( \tau_1 \right) \right] = 0,
\end{align}
where $H$ is the binary entropy function  
\begin{align}
H(\tau_1) = - \tau_1 \log_2 \tau_1 - (1-\tau_1) \log_2 (1-\tau_1), 
\end{align}
and   $\tau_2 \approx 0.344$ be the solution of
\begin{align}\label{Eq:tau2_eq}
1024\tau_2^2-384\tau_2+11= 0
\end{align}

We also let $M$ and $M'$ be the sizes of the monochromatic and almost monochromatic regions  of an arbitrary agent, respectively.

We consider  values of the intolerance $\tau \in (\tau_2,1-\tau_2) \setminus\{1/2\}$. Most of the work is devoted to the study of the intervals $(\tau_2, \tau_1]$ and $(\tau_1,1/2)$, 
 a symmetry argument  extends the analysis to the intervals $(1/2,1-\tau_1)$ and $[1-\tau_1,1-\tau_2)$. 
The following theorems show that segregation occurs for values of $\tau$ in the grey region of Figure~\ref{fig:tau}, where we expect an exponential monochromatic region, and in the black region of Figure~\ref{fig:tau}, where we expect an exponential almost monochromatic region.
 


\begin{theorem*}\label{Thrm:first_theorem}
For all $\tau \in (\tau_1,1-\tau_1) \setminus \{1/2\}$ and for sufficiently large ${\size}$, we have
\begin{align}
2^{a(\tau){\size}-o({\size})} \le \mathbb{E}[M] \le 2^{b(\tau){\size} +o({\size})},
\end{align}
where $a$ and $b$ are decreasing functions of $\tau$ for $\tau < 1/2$ and increasing for $\tau > 1/2$.
\end{theorem*}

\begin{theorem*}\label{Thrm:second_theorem}
For all $\tau \in (\tau_2,\tau_1]\cup[1-\tau_1,1-\tau_2)$ and for sufficiently large ${\size}$, we have
\begin{align}
2^{a(\tau){\size}-o({\size})} \le \mathbb{E}[M'] \le 2^{b(\tau){\size} +o({\size})},
\end{align}
where $a$ and $b$ are decreasing functions of $\tau$ for $\tau < 1/2$ and increasing for $\tau > 1/2$.
\end{theorem*}

 \begin{figure}
\centering
\includegraphics[width=3.5in]{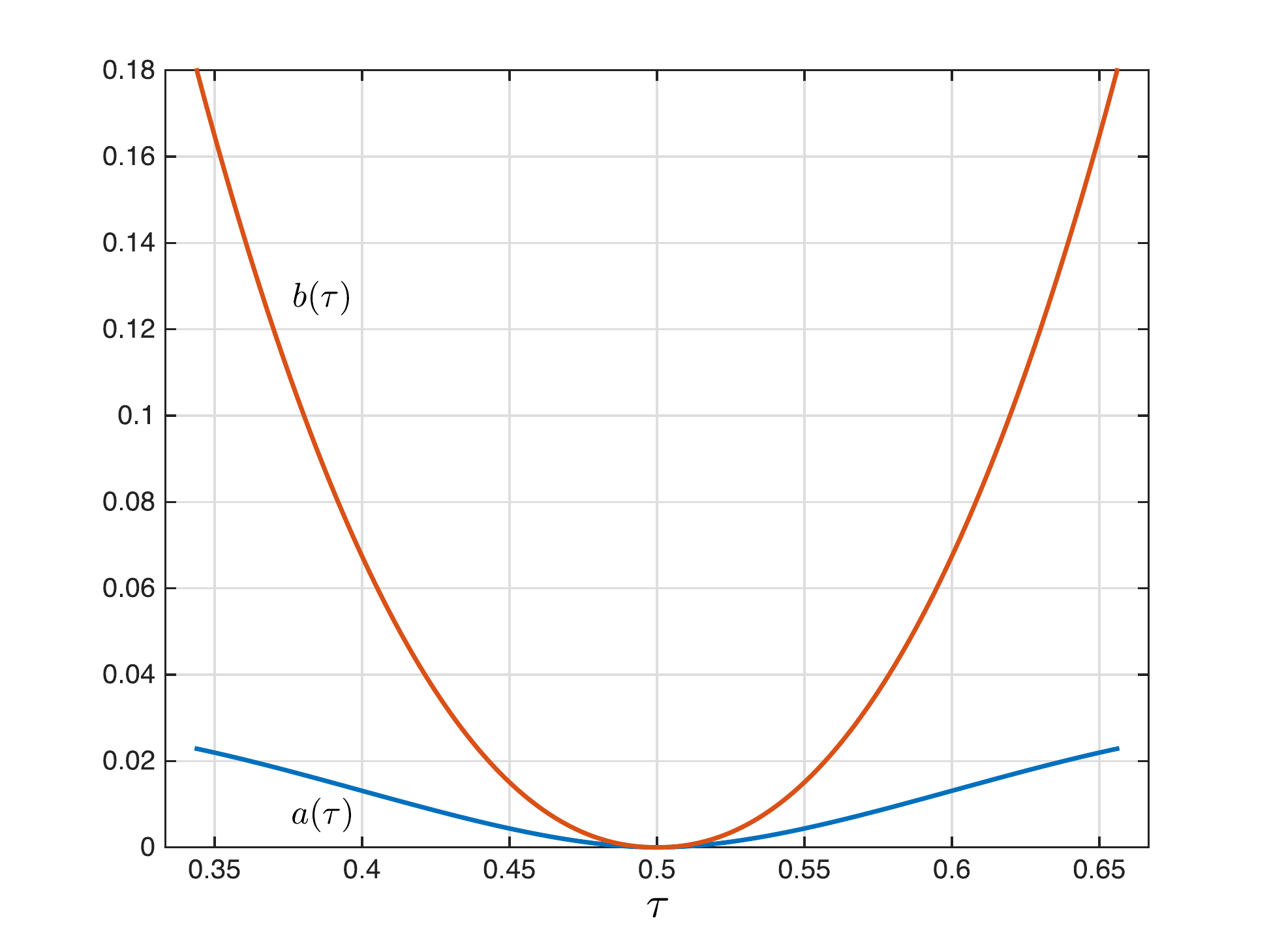}
\caption{Exponent multipliers $a(\tau)$ and $b(\tau)$ for the lower bound and upper bounds on the expected size of the largest segregated region $\mathbb{E}[M]$, and the expected size of the largest almost segregated region $\mathbb{E}[M']$. } 
\label{fig:aaprime}
\end{figure}

The numerical values for $a(\tau)$ and $b(\tau)$ derived in  the proofs of the above theorems  are plotted in Figure \ref{fig:aaprime}. For $\tau \in (\tau_1,1-\tau_1)\setminus \{1/2\}$,  as the intolerance gets farther from one half in both directions,   larger monochromatic regions are expected.

\subsection{Proof Outline} \label{Subsec:proof_const}

The main idea of the proof  is to identify a local initial configuration that can  potentially  trigger a cascading process   leading to segregation. 
We then  bound the probability of occurrence of such a configuration in the initial state, and of  the conditions  to trigger segregation.


 To identify this  local configuration, we study the relationship between  the typical neighborhood of an unhappy agent and the sub-neighborhoods contained within this neighborhood, showing a self-similar structure. Namely, the fraction of agents of the same type, when scaled by the size of the neighborhood, remains roughly the same (Proposition~\ref{Prop:firstprop}). 
We then define a \emph{radical   region}     that contains a nucleus of unhappy agents (Lemma~\ref{Lemma:unhappy_region}), and  using the self-similar structure  of the neighborhoods  we construct a geometric configuration where a sequence of flips can lead to  the formation of a neighborhood of  agents of the same type inside a radical region (Lemma~\ref{Lemma:Trigger}). Finally,
we provide a lower bound for the probability of occurrence of this configuration     in the initial state of the system (Lemma~\ref{Lemma:expandable}), which   
can initiate the segregation process.

The second part of the proof is concerned with the process dynamics, and shows a cascading effect ignited by the radical regions that leads to the formation of exponentially large segregated areas. We  consider an indestructible and impenetrable structure   around a radical region called a \textit{firewall} and show that  once formed it remains static and protects  the radical region  inside it from vanishing (Lemma~\ref{Lemma:firewall}).  Conditioned on certain events occurring in the area surrounding the radical region, including the formation of the initial configuration described in the first part of the proof, we show that an agent close to the radical region will be trapped w.h.p.\ inside an exponentially large firewall whose interior becomes monochromatic (Lemma~\ref{Lemma:fw}), see Figure~\ref{fig:aaprime1}(a).  
 We then obtain a lower bound on the  joint probability of the conditioning events and this leads to a lower bound on the probability that an agent is eventually contained in a monochromatic region of exponential size.   
 Since the lower bound holds for both type of agents, we expect to have both types of exponential monochromatic regions in a large area by the end of the process.  This leads to an exponential upper bound on the expected size of the largest  monochromatic region of each type. To perform our computations, we rely on a bound  on   the passage time on the square lattice~\cite{kesten1993speed}  to upper bound the rate of spread of other monochromatic regions  outside the firewall,  and ensure that they do not interfere with its formation during the dynamics of the process.
 
 \begin{figure}[!t] 
\centering
\includegraphics[width=\textwidth]{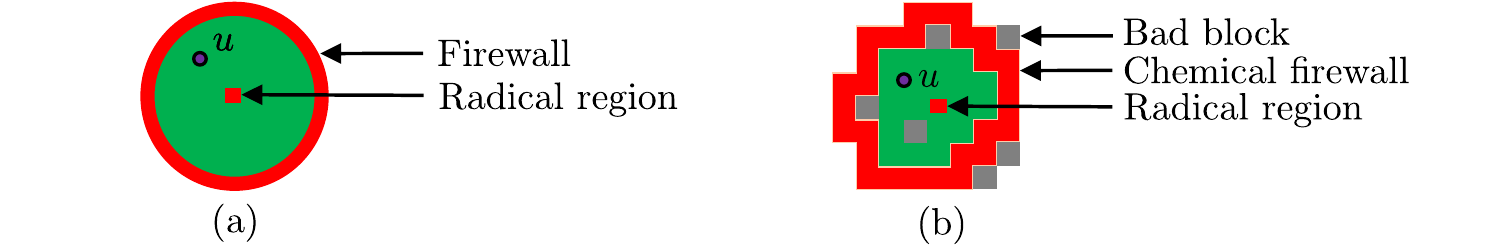}
\caption{An arbitrary agent $u$ that is close to a radical region will be trapped inside a firewall  of exponential size whose interior will eventually become   monochromatic (a), or almost monochromatic (b).} 
\label{fig:aaprime1}
\end{figure}

The construction described above works for all $\tau_1< \tau<1/2$. For smaller values of $\tau$, agents are more tolerant and this may cause the construction of a firewall to fail, since tolerant agents do not easily become unhappy and flip their types igniting the cascading process.  In order to overcome this difficulty, we introduce a  \emph{chemical  firewall}  through a comparison with a Bernoulli  site percolation model, see Figure~\ref{fig:aaprime1}(b). This firewall is constructed through renormalization and is initially  made  of  \emph{good}  blocks that occur independently and with probability above the critical threshold for site percolation on the square grid. Using a theorem  in~\cite{garet2007large} on the chemical distance between good blocks, we show that they form a large cycle that, once it becomes monochromatic,  isolates its interior. Finally, using  the exponential decay of  the size of the   clusters of bad blocks~\cite{grimmett1999percolation}, we show that
the region inside the chemical firewall  becomes \emph{almost monochromatic}, namely   for  all $\tau_2 < \tau \le \tau_1$, we expect the formation of exponentially large regions where the ratio  of number of agents of one type and the number of agents of the other type quickly vanishes.  

All results are extended to the interval  $1/2 <\tau < 1-\tau_2$ using a symmetry argument.

Compared to the proof in~\cite{immorlica2015exponential}, our derivation differs in the following aspects. The definition of radical region  is  fundamentally different from the viral nodes considered in~\cite{immorlica2015exponential}, and the identification of the radical regions gives us an   immediate understanding of the arrangement of the agents in the initial configuration in terms of  self-similarity arising at different scales.
Our definition of an annular firewall that forms quickly enough eliminates the need for additional arguments from first passage percolation that are used in~\cite{immorlica2015exponential}, it allows for a wider range of intolerance parameters, and it is easily generalized to the notion of chemical firewall using the results from \cite{garet2007large}.
The  renormalization of the grid for the study of the growth of the monochromatic regions is  also different from~\cite{immorlica2015exponential}  and works for a wider range of the intolerance. 
The  idea of considering almost monochromatic regions is   new, and so are the approaches that we use from percolation theory to argue the existence of the chemical firewall and the size of the minority clusters. Finally, we  rigorously apply a variation of the  FKG inequality to show  positive correlation of certain events,  while in~\cite{immorlica2015exponential} it is often informally argued that similar correlations exist in their setting.  

\section{Triggering configuration} \label{Sec:Triggering}
We start our analysis considering the initial configuration of the system.  
Proposition~\ref{Prop:firstprop} shows a similarity relationship between the neighborhood of an agent and its  sub-neighborhoods. This relationship is exploited in Lemma~\ref{Lemma:Trigger} to construct  an initial configuration of agents that can trigger the segregation process.  Lemma~\ref{Lemma:expandable} provides a   bound on the probability of occurrence of this triggering configuration.


Let $\mathcal{N}(u)$ be the neighborhood of an arbitrary agent $u$ containing ${\size}$ agents. Consider a sub-neighborhood $\mathcal{N}'(u) \subset \mathcal{N}(u)$ containing ${\size}'$ agents and let $\gamma$ be the scaling factor ${{\size}'}/{{\size}}$.
Let $W$ and $W'$ be the random variables representing the number of (-1) agents in $\mathcal{N}(u)$ and $\mathcal{N}'(u)$ respectively. The following proposition  shows that,  conditioned on $W$ being less than $\tau {\size}$,  $W'$ is very close to the rescaled quantity $ \gamma \tau {\size}$, with overwhelming probability as ${\size} \rightarrow \infty$.

\begin{proposition*}\label{Prop:firstprop}
For any $\epsilon \in (0,1/2)$ and $c \in \mathbb{R}^+$ there exists  $c' \in \mathbb{R}^{+}$  such that for all $\size \ge 1$
\begin{align*}
P\left(|W' -  \gamma \tau {\size}| <  c{\size}^{1/2+\epsilon} \given[\Big] W < \tau {\size}\right) \ge 1 - e^{-c'{\size}^{2\epsilon}}.
\end{align*}
\end{proposition*}

To prove this proposition, where the two constants $\epsilon$ and $c$  are introduced for technical convenience in its later applications, we need the following three lemmas.

\begin{lemma*}\label{Lemma:lemma_0}
Let $\mathcal{N}$ be a set of $(+1)$ and $(-1)$ arbitrary agents in the grid such that it has exactly $K$ agents of type $(-1)$ and ${\size}-K$ agents of type $(+1)$. Then, if we choose a set  $\mathcal{N}'$ of size ${\size}'$ of agents uniformly at random from $\mathcal{N}$, we have
\begin{align} \label{eq:Azuma1}
P(W' \ge \gamma K + t) \le e^{\frac{-t^2}{2{\size}'}} ,
\end{align}
and
\begin{align}\label{eq:Azuma2}
P(W' \le \gamma K - t) \le e^{\frac{-t^2}{2{\size}'}},
\end{align} 
where $W'$ is the random variable indicating the number of $(-1)$ agents in $\mathcal{N}'$, and $\gamma = {{\size}'}/{{\size}}$.
\end{lemma*}

\begin{proof}
Let $W'_i$ be a random variable indicating the type of the $i$'th agent in $\mathcal{N}'$, namely $W'_i$ is one if the type is (-1) and zero otherwise. Let $\mathcal{F}_i = \sigma(W'_1,...,W'_i)$, where $\sigma(X)$ denotes the sigma field generated by random variable $X$. It is easy to see that for all  $n \in \{1,...,{\size}' \}$,  $M_n = \mathbb{E}[W'|\mathcal{F}_n]$ is a martingale. It is also easy to see that $M_0 = \mathbb{E}[W'] = \gamma {\size}_\tau$, and $M_{{\size}'} = W'$. For  all $n \in\{1,2,...,{\size}'\}$,  we also have 
\begin{align*} 
|M_n - M_{n-1}|  &= \left|\mathbb{E}\left(\sum_{i=1}^{{\size}'}W'_i\given[\Big]\mathcal{F}_n\right) - \mathbb{E}\left(\sum_{i=1}^{{\size}'}W'_i\given[\Big]\mathcal{F}_{n-1}\right)\right|  \\
&=  \left|W'_n + \frac{K-\sum_{i=1}^{n}W'_i}{{\size}-n}({\size}'-n) 
 - \frac{K-\sum_{i=1}^{n-1}W'_i}{{\size}-(n-1)}[{\size}'-(n-1)]\right| \\
&\le 1.
\end{align*}
Now, using Azuma's inequality \cite{janson2011random}, we have
\begin{align*}
P\left(W'_i \ge \gamma K + t\right) = P\left(M_{{\size}'} \ge M_0 + t)\right) \le e^{\frac{-t^2}{2{\size}'}}.
\end{align*} 
With the same argument we can derive (\ref{eq:Azuma2}).   
\end{proof}

\begin{lemma*}\label{Lemma:lemma_1}
Let $\epsilon \in (0,1/2)$ and $c \in \mathbb{R}^+$. There exists $c' \in \mathbb{R}^{+}$ such that for all $\size \ge 1$
\begin{align*}
{P\left(W' < \gamma \tau {\size}  + c{\size}^{1/2+\epsilon} \given[\Big] W < \tau {\size}\right) \ge 1 - e^{-c'{\size}^{2\epsilon}}}.
 \end{align*}
\end{lemma*}

\begin{proof}
Let us denote $c{\size}^{1/2+\epsilon}$ by $v({\size})$. We let
\begin{align*}
p_w &= P\left(W' \ge\gamma \tau {\size} + v({\size}) \given[\Big] W < \tau {\size}\right) \\
&\le P\left(W' \ge \gamma \tau {\size} + v({\size}) \given[\Big] W \le \tau {\size}\right) \\
&\le P\left(W' \ge \gamma \tau {\size} + v({\size})\given[\Big] W = \tau {\size}\right) 
\end{align*}
The first inequality is trivial. The second inequality follows from   
\begin{align*}
P\left(W' \ge \gamma \tau {\size} + v({\size}) \given[\Big]W \le \tau {\size}\right) 
\end{align*} being the probability of choosing $W' \ge \gamma \tau {\size} + v({\size})$ agents from a set with $W \le \tau {\size}$. It is easy to see that this probability can only increase if we have $W = \tau {\size}$.
The result follows by applying Lemma~\ref{Lemma:lemma_0}. 
\end{proof}

Let $\mathcal{N}''(u) = \mathcal{N}(u) \setminus \mathcal{N}'(u)$. Let us denote the number of agents in $\mathcal{N}''(u)$ by ${\size}''$. Let $W''$ denote the random variable representing the number of (-1) agents in $\mathcal{N}''(u)$.
\begin{lemma*}\label{Lemma:lemma_2}
Let $\epsilon \in (0,1/2)$ and $c \in \mathbb{R}^+$. There exist  $c' \in \mathbb{R}^{+}$ such that for all $\size \ge 1$ 
\begin{align*}
{P\left(W' > \gamma \tau {\size}  - c{\size}^{1/2+\epsilon} \given[\Big] W < \tau {\size}\right) \ge 1 -  e^{-c'{\size}^{2\epsilon}}}. 
\end{align*}
\end{lemma*}

\begin{proof}
Let us denote $c{\size}^{1/2+\epsilon}$ by $v({\size})$, and $ \tau {\size}  - 1$ by ${\size}_\tau$. Let
\begin{align}
p_w &= P\left(W' \le \tau\gamma {\size} - v({\size}) | W < \tau {\size}\right) \nonumber \\
&= P\left(W' \le \tau {\size}' - v({\size}) | W' + W'' < \tau {\size}\right) \nonumber \\
&\le \ddfrac{P\left(W' \le \tau {\size}' - v({\size}), W'+W'' \le {\size}_\tau \right)}{P(W \le {\size}_\tau)} \nonumber \\
&\le \ddfrac{\sum\limits_{k=0}^{\lfloor \tau {\size}' - v({\size}) \rfloor}P(W' = k)\sum\limits_{m=0}^{\min\{{\size}_\tau - k,{\size}''\}}P(W''=m)}{P(W \le {\size}_\tau)} \nonumber \\
&= \ddfrac{\sum\limits_{k=0}^{\lfloor \tau {\size}' - v({\size}) \rfloor}{{\size}' \choose k}\sum\limits_{m=0}^{\min\{{\size}_\tau - k,{\size}''\}}{{\size}'' \choose m}}{\sum\limits_{n=0}^{{\size}_\tau}{{\size} \choose n}}. 
\label{denominator}
\end{align}
We use the following inequality, valid for all     $a \in (0,0.5)$ 
\begin{equation*}
{{{\size} \choose a{\size}} \le\sum_{m=0}^{a{\size}}{{\size} \choose m} \le \frac{1-a}{1-2a}{{\size} \choose a{\size}}}.
\end{equation*}
Since $\tau < 1/2$, it follows that ${\size} \choose {\size}_\tau$ is a lower bound for the denominator of (\ref{denominator}). We  also have the following upper bound for the numerator  
\begin{align*}
\sum\limits_{k=0}^{\lfloor \tau {\size}' - v({\size}) \rfloor} & {{\size}' \choose k}\sum\limits_{m=0}^{\min\{{\size}_\tau - k,{\size}''\}}{{\size}'' \choose m}  \le \\ 
&\sum\limits_{k=0}^{\lfloor \tau {\size}' - v({\size}) \rfloor}c^k{{\size}' \choose k}{{\size}'' \choose {\min\{{\size}_\tau - k,\lfloor {\size}''/2 \rfloor \}}},
\end{align*}
where $\{c^k\}$ are positive constants for $k=0,1,...,\lfloor \tau {\size}' - v({\size}) \rfloor$. Since for all $l \in \{0,1,...,\lfloor \tau {\size}' - v({\size}) \rfloor\}$, we have
\begin{align*}
\ddfrac{{{\size}' \choose \lfloor \tau {\size}' - v({\size}) \rfloor}{{\size}'' \choose {\min\{{\size}_\tau - \lfloor \tau {\size}' - v({\size}) \rfloor,\lfloor {\size}''/2 \rfloor \}}}}{{{\size}' \choose \lfloor \tau {\size}' - v({\size}) \rfloor - l}{{\size}'' \choose {\min\{{\size}_\tau - \lfloor \tau {\size}' - v({\size}) \rfloor + l,\lfloor {\size}''/2 \rfloor \}}}} \ge 1,
\end{align*}
it follows that there exist a constant $c_1 \in \mathbb{R}^+$ such that
\begin{align*}
c_1{\size}{{\size}' \choose \lfloor\tau {\size}' - v({\size})\rfloor}{{\size}'' \choose {\size}_\tau  - \lfloor\tau {\size}' - v({\size})\rfloor}
\end{align*}
is an upper bound for the numerator. Putting things together,  we have
\begin{align*}
p_w &\le c_1{\size}\ddfrac{{{\size}' \choose \lfloor\tau {\size}' - v({\size})\rfloor}{{\size}'' \choose {\size}_\tau  - \lfloor\tau {\size}' - v({\size})\rfloor}}{{{\size} \choose {\size}_\tau}} \\
&\le c_1{\size} P(W'\le \tau {\size}'-v({\size})|W= {\size}_\tau ). 
\end{align*}
Using the same argument as in Lemma \ref{Lemma:lemma_1}, we now have 
\begin{align*}
p_w \le e^{-c'{\size}^{2\epsilon}},
\end{align*}
where $c' \in \mathbb{R}^+$ is a constant.  
\end{proof}

\begin{proof}[{Proposition~\ref{Prop:firstprop}}] 
Let 
\begin{align*}
A=\left\{\tau\gamma {\size} - c{\size}^{1/2+\epsilon} < W'\right\}, \\
B = {\left\{W' < \tau\gamma {\size} + c{\size}^{1/2+\epsilon}\right\}}, \\
C=\left\{W < \tau {\size}\right\}.
\end{align*}
By Lemmas \ref{Lemma:lemma_1} and \ref{Lemma:lemma_2} there exist constants $c_1,c_2>0$ such that we have
\begin{align*}
P(A\cap B|C) &= 1 -  P\left(A^C\cup B^C\given[\Big]C\right) \\
&\ge 1 - \left(P\left(A^C\given[\Big]C\right)+P\left(B^C\given[\Big]C\right)\right) \\
&\ge 1 - \left(e^{-c_1{\size}^{2\epsilon}}+e^{-c_2{\size}^{2\epsilon}}\right).
\end{align*}
Hence, there exists a constant $c' \in \mathbb{R}^+$ such that
\begin{align*}
P\left(A\cap B\given[\Big]C\right) \ge 1 - e^{-c'{\size}^{2\epsilon}},
\end{align*}
and the proof is complete.  
\end{proof}


We now identify a configuration that has the potential to trigger a cascading process. 
We show that a neighborhood that is slightly larger than the neighborhood of an agent and that contains a fraction of same type agents that is slightly less than $\tau$ has the desired configuration.
For any $\epsilon, \epsilon' \in {(0,1/2)}$ let ${\hat{\tau} = \tau [1- 1/ (\tau {\size}^{1/2-\epsilon})]}$ and define  a
\textit{radical region} $\mathcal{N}_{(1+\epsilon')w}$  to be a neighborhood 
of radius ${(1+\epsilon')w}$ containing  less than $\hat{\tau}(1+\epsilon')^2{\size} $  agents of type (-1). We also define
an \textit{unhappy region} $\mathcal{N}_{\epsilon'w}$ to be a neighborhood of radius $\epsilon' w$, containing at least $\lfloor \tau \epsilon'^2 {\size}-{\size}^{1/2+\epsilon} \rfloor$ unhappy agents of type  (-1). 
\begin{lemma*} \label{Lemma:unhappy_region}
A radical region $\mathcal{N}_{(1+\epsilon')w}$ contains an unhappy region $\mathcal{N}_{\epsilon'w}$ at its center w.h.p.
\end{lemma*}

\begin{proof}
Let $\epsilon \in (0,1/2)$. 
We show that w.h.p. the region $\mathcal{N}_{\epsilon'w}$ co-centered with $\mathcal{N}_{(1+\epsilon')w}$ has at least ${\lfloor \tau\epsilon'^2{\size} - {\size}^{1/2+\epsilon}}\rfloor$ agents of type (-1) such that all of them are unhappy. 
Let $A$ be the event that there are less than $\tau \epsilon'^2 {\size} - {\size}^{1/2+\epsilon}$ agents of type (-1) in $\mathcal{N}_{\epsilon'w}$, which has $N'$ agents. By Proposition~\ref{Prop:firstprop}, there exists $c_1, c_2>0$ such that
\begin{align*}
P(A) \le P\left(W'\le\hat{\tau} {\size}' - c_1{\size}^{1/2+\epsilon}\given[\Big] W_{(1+\epsilon')w} < (1+\epsilon')^2\hat{\tau} {\size}\right) \le e^{-c_2{\size}^{2\epsilon}},
\end{align*}
 where $W_{(1+\epsilon')w}$ represents the number of (-1) agents in $\mathcal{N}_{(1+\epsilon')w}$. Let $\mathcal{I}$ denote the set of the positions of all the agents in $\mathcal{N}_{\epsilon'w}$, and let $B_i$ be the event that a (-1) agent positioned at $i\in \mathcal{I}$ is happy. By Proposition \ref{Prop:firstprop}, there exists $c_3>0$ such that, for all $i\in\mathcal{I}$ 
 \begin{align*}
 P(B_i) = P\left(W_i\ge \hat{\tau} {\size} + c_u{\size}^{1/2+\epsilon}\given[\Big]W_{(1+\epsilon')w} < (1+\epsilon')^2\hat{\tau} {\size}\right) \le e^{-c_3{\size}^{2\epsilon}}, 
 \end{align*}
 where $W_i$ is the number of (-1) agents in the neighborhood of $i$ and $c_u>0$ is chosen so that the threshold for being happy is met. It follows that there exists $c>0$ such that
\begin{align*}
 P\left(A\cap B_1^C \cap ... \cap B_{|\mathcal{I}|}^C\right) \ge 1 - {\size}e^{-c{\size}^{2\epsilon}}, 
 \end{align*}
where $|\mathcal{I}|$ denotes the cardinality of $\mathcal{I}$.  
\end{proof}

A radical region is \emph{expandable} if there is a sequence of  at most $(w+1)^2$ possible flips inside it that can make  the neighborhood $\mathcal{N}_{w/2}$ at its center monochromatic.

We   consider a geometric configuration 
where a radical region,  and  neighborhoods $\mathcal{N}_{\epsilon'w}$ , $\mathcal{N}_{w/2}$ and $\mathcal{N}_{{\radius}}$  with ${\radius}>3w$, are all co-centered. 
We consider the process dynamics and let $u^+$ denote an arbitrary (+1) agent and 
\begin{align}
T({\radius}) = \inf\left\{t: \exists v\in \mathcal{N}_{\radius}, \; u^{+}  {\normalfont  \mbox{ would be unhappy at the location of $v$} } \right\}.
\label{Tinfdef}
\end{align}

The next lemma shows that   the radical region in this configuration     is expandable w.h.p., provided that  $\epsilon'$ is large enough and no (+1) agent at the location of any agent in $\mathcal{N}_{\radius}$ is unhappy.  The main idea is that the (-1) agents in the  unhappy region at the center of the radical region can  trigger a process that leads to a monochromatic (+1) region of radius $w/2$.


\begin{lemma*} \label{Lemma:Trigger}
For all $\epsilon' > f(\tau)$, where
\begin{align}
f(\tau) = \frac{3(\tau - 0.5)+\sqrt{9(\tau-0.5)^2-7(\tau-0.5)(3\tau+0.5)}}{2(3\tau+0.5)},
\label{eq:ftau}
\end{align}
there exists w.h.p. a sequence of at most $(w+1)^2$ possible flips  in $\mathcal{N}_{(1+\epsilon')w}$ such that if they  happen  before $T({\radius})$,  then all the agents inside   $\mathcal{N}_{w/2}$ will become of  type (+1).
\end{lemma*}

\begin{figure}[!t]
\centering
\includegraphics[width=\textwidth]{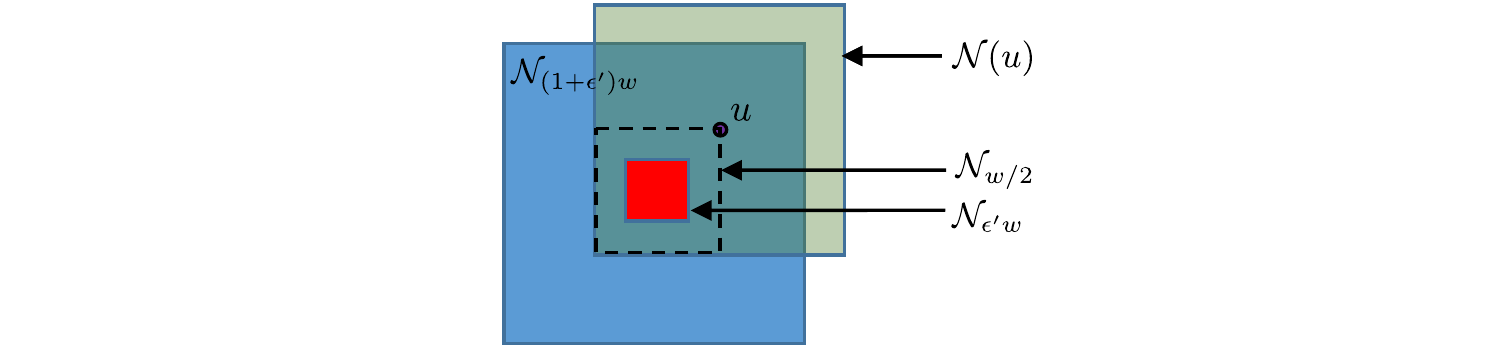}
\caption{Regions discussed in Lemma~\ref{Lemma:Trigger}. $\mathcal{N}_{\epsilon'w}$ is an unhappy region w.h.p., the dashed box is $\mathcal{N}_{w/2}$, $u$ is a corner agent in $\mathcal{N}_{w/2}$, and finally $\mathcal{N}(u)$ is the neighborhood of agent $u$.}
\label{fig:thrm41s}
\end{figure}

\begin{proof}
Let $\epsilon \in (0,1/2)$. Let us denote the neighborhood with radius $\epsilon'w$ and co-centered with the radical region by $\mathcal{N}_{\epsilon'w}$, see Figure \ref{fig:thrm41s}. By Lemma \ref{Lemma:unhappy_region}, with probability at least $1-e^{-O({\size}^{2\epsilon})}$ there are at least $\lfloor \tau\epsilon'^2{\size} - {\size}^{2\epsilon} \rfloor$ agents of type (-1) inside this neighborhood such that all of them are unhappy. Next, we  show that if these unhappy agents flip before $T(\radius)$, all the agents inside the neighborhood $\mathcal{N}_{w/2}$ will be unhappy w.h.p., which gives the desired result. 

First, we notice that if there is a flip of an unhappy (-1) agent in $\mathcal{N}_{\radius}\setminus \mathcal{N}_{w/2}$ it can only increase the probability of the existence of the sequence of flips we are looking for, hence conditioned on having these flips before $T(\radius)$, the worst case is when these flips occur with the initial configuration of $\mathcal{N}_{\radius}\setminus \mathcal{N}_{w/2}$. 
Since a corner agent in $\mathcal{N}_{w/2}$ shares the least number of agents with the radical region, it is more likely for it to have the largest number of (+1) agents in its neighborhood compared to other agents in $\mathcal{N}_{w/2}$. Hence, as a worst case, we may consider a corner agent in $\mathcal{N}_{w/2}$ which is co-centered with the radical region.

Let us assume that $\epsilon' \in (0,1/2)$, in this case $\mathcal{N}_{\epsilon'w}$ is  completely contained in the neighborhood of each of the agents in   $\mathcal{N}_{w/2}$. Let us denote the neighborhood shared between the neighborhood of the agent $u$ at the corner of $\mathcal{N}_{w/2}$ and the radical region by $\mathcal{N}''(u)$. Also, let us denote the  scaling factor corresponding to this shared neighborhood by $\gamma''$. We have
 \begin{align*}
 \gamma''= \frac{(3/2+\epsilon')^2}{4(1+\epsilon')^2} \pm O\left(\frac{1}{\sqrt{{\size}}}\right).
 \end{align*}
By Proposition \ref{Prop:firstprop} it follows that with probability at least ${1-e^{-O({\size}^{2\epsilon})}}$ there are at most
  \begin{align*}
\frac{(3/2+\epsilon')^2\tau}{4}{\size} + o(N),
  \end{align*}
agents of type (-1) in $\mathcal{N}''(u)$. Hence, we can conclude that, for any agent in $\mathcal{N}_{w/2}$, w.h.p., there are at most this many (-1) agents in the intersection of the neighborhood of this agent and the radical region.

Also, using Lemma \ref{Lemma:balanced} of the Appendix, with probability at least $1-e^{-O(N^{2\epsilon})}$ we have at most
  \begin{align*}
  { \frac{1}{2}\left(1-(3/2+\epsilon')^2/4\right){\size} + o({\size}) },
  \end{align*}
 agents of type (-1) in the part of the neighborhood of the corner agent $u$  in $\mathcal{N}_{w/2}$ that is also not  in the radical region. Combining the above results, we can conclude that with probability at least $1-e^{-O(N^{2\epsilon})}$ there are at most 
\begin{align*}
\frac{(3/2+\epsilon')^2\tau}{4}{\size} + \frac{1}{2}\left(1-\frac{(3/2+\epsilon')^2\tau}{4}\right){\size}  + o({\size}),
 \end{align*} 
agents of type (-1) in the neighborhood of an agent in $\mathcal{N}_{w/2}$. Let us denote this event for the corner agent~$u$ by $A_1$. Let us denote the events of having at most this many (-1) agents in the neighborhoods of other agents in $\mathcal{N}_{w/2}$ by $A_2, ..., A_{|\mathcal{N}_{w/2}|}$, where $|\mathcal{N}_{w/2}|$ denotes the number of agents in $\mathcal{N}_{w/2}$. We have
\begin{align*}
P(A_1 \cap ... \cap A_{(w+1)^2}) &\ge 1 - P(A_1^C \cup ... \cup A^C_{|\mathcal{N}_{w/2}|}) \\
&\ge 1 - (w+1)^2 P(A^C_1) \\ 
& \ge 1 - e^{-O(N^{2\epsilon})}.
\end{align*}


The goal is now to find the range of $\epsilon'$ for which $\mathcal{N}_{\epsilon'w}$ is large enough that once all of its unhappy agents flip, all the agents in $\mathcal{N}_{w/2}$ become unhappy w.h.p. 
It follows that we need 
\begin{align*}
\frac{(3/2+\epsilon')^2\tau}{4}{\size} + \frac{1}{2}\left(1-\frac{(3/2+\epsilon')^2\tau}{4}\right){\size} - \tau\epsilon'^2{\size} + o({\size}) < \tau {\size},
 \end{align*} 
to hold w.h.p. Dividing by ${\size}$, and letting ${\size}$ go to infinity,  after some   algebra it follows that
 \begin{align}\label{Eq:eps1}
\epsilon' > \frac{3(\tau - 0.5)+\sqrt{9(\tau-0.5)^2-7(\tau-0.5)(3\tau+0.5)}}{2(3\tau+0.5)} = f(\tau),
\end{align}
where $f(\tau) < 1/2$ for $\tau \in (\tau_2,1/2)$, as desired.  
\end{proof}

\begin{figure}[!t]
\centering
\includegraphics[width=4.5in]{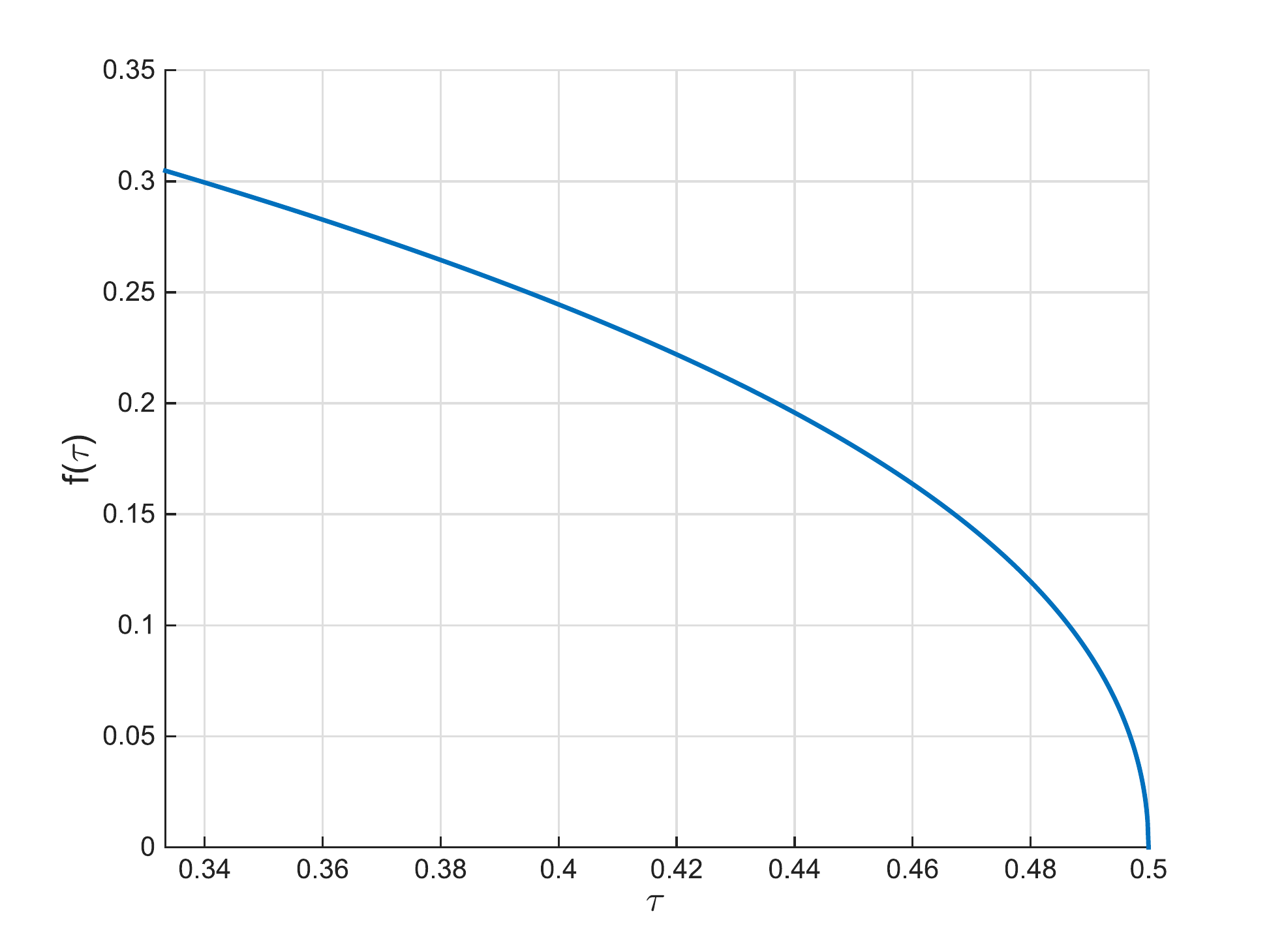}
\caption{The infimum of $\epsilon'$ to potentially trigger a cascading process.}
\label{fig:epsilon}
\end{figure}

Figure \ref{fig:epsilon} depicts $f(\tau)$ as a function of $\tau$. When $\tau$ is close to one half, it is sufficient to have an $\epsilon'$ close to zero to potentially trigger a segregation process. In this case, a small number of agents located in a small unhappy region are needed to flip in order to make other agents in the radical region unhappy. However, as $\tau$ decreases and agents become more tolerant, a larger number of agents must make a flip in the unhappy region in order to make other agents in the radical region unhappy, and hence larger values of $\epsilon'$ are needed. 


Using Lemma~\ref{Lemma:Trigger}, we obtain  an exponential bound  
  on the probability of having an expandable radical region inside a sufficiently large neighborhood.  This shows that the probability that an expandable radical region is sufficiently close to an arbitrary agent $u$ in the initial configuration, is not too small. 
 
\begin{lemma*} \label{Lemma:expandable}
Let $r = 2^{[1-H(\tau')]{\size}/2-o({\size})}$, where $\tau'=(\tau {\size} -2)/({\size}-1)$. Let
\begin{align*}
C = \left\{\text{$\mathcal{N}_r$  \emph{contains an expandable radical region at} $t=0$}\right\}.
\end{align*}
For all $\epsilon' > f(\tau)$ and sufficiently large ${\size}$, we have
\begin{align*}
P(C) \ge 2^{-[1-H(\tau')](2\epsilon'+\epsilon'^2){\size} -o({\size})}.
\end{align*}
\end{lemma*} 
\begin{proof}
Let $\mathcal{N}_r$ be an arbitrary neighborhood of radius $r= 2^{[1-H(\tau')]{\size}/2-o({\size})}$ and let $\mathcal{N}_{\radius}$ be a neighborhood of radius $\radius = r + w$ and with the same center as $\mathcal{N}_r$. Let
\begin{align*}
A = \{\forall v\in \mathcal{N}_{\radius}, \;  u^{+}   \mbox{ would be happy at the location of } v \mbox{ at time } t=0\}, 
\end{align*}
\begin{align*}
C = \text{\{$\mathcal{N}_r$ {\normalfont contains an expandable radical region at time $t=0$\}}},
\end{align*}
\begin{align*}
S_{\epsilon'} = \text{\{$\mathcal{N}_r$ contains a radical region of radius $(1+\epsilon')w$ at time $t=0$\}}.
\end{align*}
We have
\begin{align*}
P(C) &\ge P(C\cap S_{\epsilon'} \cap A) \\
&= P\left(C\given[\Big]A,S_{\epsilon'}\right)P(S_{\epsilon'}\cap A).
\end{align*}
Using the FKG inequality and since $S_{\epsilon'}$ and $A$ are increasing events,  we have
\begin{align*}
P(C) \ge P\left(C\given[\Big]A,S_{\epsilon'}\right)P(S_{\epsilon'})P(A).
\end{align*}
By Lemma~\ref{Lemma:Trigger} we have that $P(C|A,S_{\epsilon'})$ occurs w.h.p.  By Lemmas  \ref{Lemma:R_unhappy}  and \ref{Lemma:r_unhappy} of the Appendix we have that 
\begin{align*}
P(S_{\epsilon'}) \ge 2^{-[1-H(\tau')][2\epsilon'+\epsilon'^2]{\size} -o({\size})}.
\end{align*} 
Finally, $P(A)$ tends to one as ${\size} \rightarrow \infty$ which leads to the desired result.  
\end{proof}

So far,  we have identified a local configuration (radical region) that can   lead to the formation of a small monochromatic neighborhood w.h.p. In the following section we show that this monochromatic neighborhood is in fact capable of making a large region monochromatic or almost monochromatic. 

\section{The segregation process}  \label{Sec:Segregation}
We now consider the dynamics of the segregation process and show that for all $\tau \in (\tau_1, 1/2)$ the expected size of the  monochromatic region  in steady state is exponential, while for all $\tau \in (\tau_2,\tau_1]$ the expected size of the almost monochromatic region is exponential. 

\subsection{Monochromatic region}
We need the following definitions and preliminary results for proving the first part of Theorem~\ref{Thrm:first_theorem}.
A \textit{firewall} of radius $r$ and center $u$ is a set of agents of the same type  contained in an annulus 
\begin{align*}
A_r(u) = \left\{y: r-\sqrt{2}w \leq \|u-y\| \leq r\right\},
\end{align*}
where $\|.\|$ denotes  Euclidean distance and $r\ge 3w$. By Lemma \ref{Lemma:firewall}, once formed a firewall of sufficiently large radius  remains  static,  and since its width is  $\sqrt{2} w$ the agents inside the inner circle are not going to be affected by the configurations outside the firewall.

We now call a neighborhood with radius $w/2$ a $w$-block. Consider the grid graph $G_n$. Let us renormalize this grid into $w$-blocks and denote the resulting graph by $G'_n$ where each vertex of it is a $w$-block. Consider i.i.d. random variables $\{t(v):v\in G'_n\}$, each attached to a vertex of $G'_n$. Let $F$ denote the common distribution of these random variables and assume $F(0^{-}) = 0$, $\int_{[0,\infty)}xF(dx) < \infty$, and that $F$ is not concentrated on one point. Consider a path $\eta$ consisting of the vertices $v_1,...,v_k \in G'_n$ and  define the passage time of  this path 
\begin{equation*}
T^*(\eta) = \sum_{i=1}^k t(v_i).
\end{equation*}
We also define
\begin{align*}
T_k &= \inf_{\eta \in (0 \leftrightarrow k\zeta_1)} \{T^*(\eta) \},
\end{align*}
 where $\zeta_1$ is a coordinate vector and $(0 \leftrightarrow k\zeta_1)$ indicates the set of paths between the origin and $k\zeta_1$.

The following theorem,  originally stated for  bond percolation, also holds for   site percolation and appears as Theorem~1 in \cite{kesten1993speed}.

\begin{theorem*} [Kesten] \label{Thrm:Kesten_thrm1}
Let  $F(0) < p_c(\mathbb{Z}^d)$ where $p_c$ is the critical probability for site percolation on $\mathbb{Z}^d$, and $\int e^{\gamma x}F(dx)<\infty$ for some $\gamma > 0$. Then, there exist $c_1,c_2,c_3,c_4 \in \mathbb{R}^+$ independent of $k$ and such that
\begin{align*}
{P\left(|T_k - \mathbb{E}[T_k]|>x\sqrt{k}\right)<c_1e^{-c_2x}}, 
\end{align*} 
 for $x<c_3k$ and $c_4k^{-2} \le \mathbb{E}[T_k]/{k}-\mu$ where $\mu = \lim_{k\rightarrow \infty} {T_k}/{k}$.
\end{theorem*}

Using the above theorem, we obtain the following lower bound on the conditional probability that the spread of unhappy agents takes a sufficiently large amount of time.
\begin{lemma*} \label{Lemma:Unhappy_growth}
Let $\mathcal{N}_{\radius}$ be a neighborhood with radius ${\radius}>{\size}^3$ and let $u^+$ denote an arbitrary (+1) agent. 
Let 
\begin{align*}
A = \left\{\forall v\in \mathcal{N}_{\radius}, \;  u^{+}   \mbox{ {\normalfont would be happy at the location of $v$ at time $t=0$}}\right\}.
\end{align*}
There exist constants $c,c',c'' \in \mathbb{R}^+$ independent of $N$, such that for all ${\size \ge 1}$,
\begin{align*}
P\left(T(\radius/2)>c''\frac{{\radius}}{{\size}^{3/2}}\given[\Big] A\right) > 1 - {c{\radius}^2}e^{-c' {\radius}^{1/3}},
\end{align*}
where $T(\radius)$ is defined in (\ref{Tinfdef}). 
\end{lemma*}

\begin{figure}[!t]
\centering
\includegraphics[width=\textwidth]{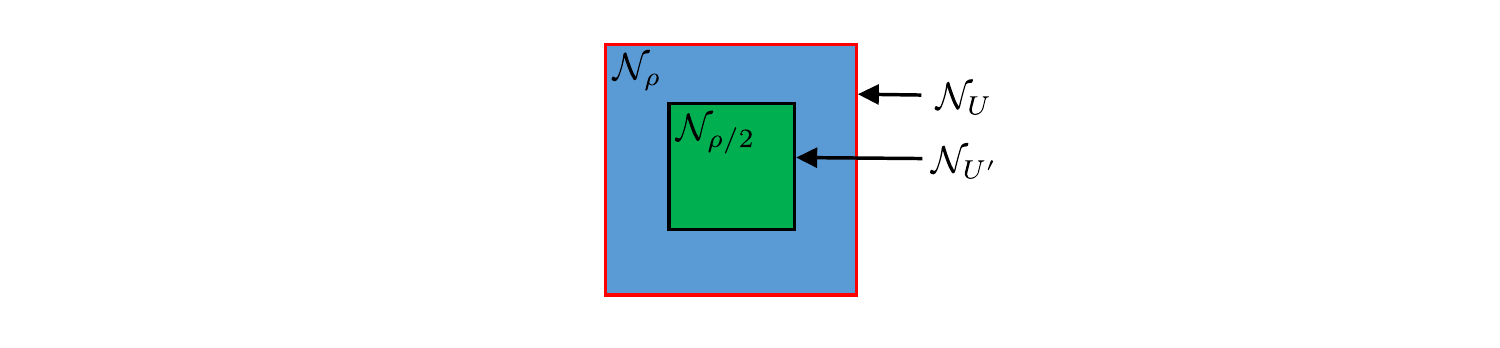}
\caption{Neighborhoods described in the proof of Lemma~\ref{Lemma:Unhappy_growth}.}
\label{fig:Unhappy_growth}
\end{figure}

\begin{proof}
We renormalize the grid into $w$-blocks starting with the block at the center of $\mathcal{N}_{\radius}$ and construct $G'_n$ as described above. Let $\mathcal{N}_U$ be the set of all the $w$-blocks  on the outside boundary of $\mathcal{N}_{\radius}$ (these are the blocks that are connected to $\mathcal{N}_{\radius}$ in $G'_n$). In order to find an upper bound for the speed of the spread of the unhappy agents, assume that all the (+1) agents in a $w$-block will become unhappy with a single flip in one of its eight $l_\infty$ closest neighboring $w$-blocks. Also assume that all the agents in $\mathcal{N}_U$ are unhappy of type (+1). Finally, denote the $w$-blocks on the outside boundary of $\mathcal{N}_{{\radius}/2}$ with $\mathcal{N}_{U'}$.


We show that the speed of the spread of unhappy blocks, i.e., $w$-blocks containing unhappy agents, is independent of the configuration of the agents outside the neighborhood $\mathcal{N}_{\radius} \cup \mathcal{N}_U$ and then  use Theorem~\ref{Thrm:Kesten_thrm1} to obtain the final result.

 Consider $G'_n$ in which each vertex is a $w$-block as described above. 
Here we attach i.i.d. random variables $\{t(v):v\in G'_n\}$ to each vertex. Let these random variables have a common exponential distribution with mean $1/{\size}$. 
 Consider a path $\eta$ consisting of the verticies $v_1,...,v_k$ and the passage time $T^*(\eta) = \sum_{i=1}^k t (v_i)$. Let 
\begin{align*} 
T' = \inf_{\eta \in  (\mathcal{N}_U \leftrightarrow   \mathcal{N}_{U'})} T^*(\eta),
\end{align*} 
where $(\mathcal{N}_U \leftrightarrow   \mathcal{N}_{U'})$ is the set of paths connecting $\mathcal{N}_U$ to  $\mathcal{N}_{U'}$.
It is easy to see that $T' \le T({\radius}/2)$.

We now argue that regardless of the configuration of agents in the blocks of the graph $G'_n$ containing $\mathcal{N}_{\radius} \cup \mathcal{N}_U$, the path with the smallest $T^*(\eta)$ consists only of $w$-blocks inside $\mathcal{N}_{\radius} \cup \mathcal{N}_U$. Assume that this is not the case, then a $w$-block is in $T^*(\eta)$ but it is not in $\mathcal{N}_{\radius} \cup \mathcal{N}_U$. There needs to be a path from this block to a block in $\mathcal{N}_{U'}$.  This path has to cross the $\mathcal{N}_U$, and as a result there is another path from $\mathcal{N}_U$ to $\mathcal{N}_{U'}$ that is at least as short as $\eta$. It follows that the shortest path from $\mathcal{N}_U$ to $\mathcal{N}_{U'}$ only consists of blocks from $\mathcal{N}_{\radius}$. 

Now we can assume that $\mathcal{N}_{\radius} \cup \mathcal{N}_U$  is in an infinite lattice of blocks $\mathbb{L}$, where i.i.d. random variables $\{t(v):v\in\mathbb{L}\}$ are attached to its nodes.
Let $B_{U}$ and $B_{U'}$ be two blocks in $\mathcal{N}_{U}$ and $\mathcal{N}_{U'}$ that have the minimum $l_1$ distance. We let
\begin{align*}
T'' = \inf_{\eta \in  (B_U \leftrightarrow   B_{U'})} T^*(\eta).
\end{align*} 
By Theorem~\ref{Thrm:Kesten_thrm1} and  since the neighborhood is divided into $w$-blocks so that $k$ is proportional to ${\radius}/\sqrt{{\size}}$, we conclude that there exist a constant $c'' \in \mathbb{R}^+$ such that for any pair of $w$-blocks in $\mathcal{N}_U$ and $\mathcal{N}_{U'}$, there exist constants $c,c' \in \mathbb{R}^+$ such that for all ${\size \ge 1}$
\begin{align*}
P\left(T'' \le c''\frac{{\radius}}{{\size}^{3/2}} \given[\Big] A \right) &\le P\left(T'' \le \frac{{\radius}}{{\size}^{1/2}}\frac{\mu}{{\size}}-x\sqrt{\frac{{\radius}}{\sqrt{{\size}}}} \given[\Big] A\right) \\
&\le P\left(T'' \le \mathbb{E}[T'']-x\sqrt{\frac{{\radius}}{\sqrt{{\size}}}} \given[\Big] A \right) \\
& \le  ce^{-c'{({\radius})^{1/3}}},
\end{align*} 
where  $x={\radius}^{1/3}$ and we have used the fact that if for a  first passage percolation process with exponential distribution with unit mean we have $\lim_{n\rightarrow \infty}{T_n}/{n} = \mu$, then for the passage times of our process, which is assumed  to be exponential with mean ${1}/{{\size}}$, we have $\lim_{n\rightarrow \infty} {T_n}/{n} = {\mu}/{{\size}}$. Finally, by the union bound, the probability that any of the unhappy agents in $\mathcal{N}_U$ affects an agent in $\mathcal{N}_{U'}$ before or at time $c''{{\radius}}/{{\size}^{3/2}}$ is at most $c{(4{\radius})}{(8{\radius})}e^{-c'({\radius})^{1/3}}$. Hence, we have 
\begin{align*}
P\left(T({\radius}/2)>c''\frac{{\radius}}{{\size}^{3/2}}\given[\Big] A\right) &\ge P\left(T'>c''\frac{{\radius}}{{\size}^{3/2}}\given[\Big] A\right) \\
&> 1 - c{(4{\radius})}{(8{\radius})}e^{-c'({\radius})^{1/3}},
\end{align*}
which tends to one as $\size \rightarrow \infty$.   
\end{proof}

Call  a \emph{region of expansion}   any neighborhood whose configuration is such that by placing a neighborhood $\mathcal{N}_{w/2}$ of type (+1) agents anywhere inside it, all the (-1) agents on the outside boundary of $\mathcal{N}_{w/2}$   become unhappy with probability one.

\begin{lemma*} \label{Lemma:monoch_spread_1}
Let $\tau \in (\tau_1,1/2)$ and let $\mathcal{N}_{4r}$ be a neighborhood of radius $4r=2^{[1-H(\tau')]{\size}/2-o({\size})}$ such that ${\radius} > 8r$. Let
\begin{align*}
D = \left\{\text{$\forall t < T({\radius}/2), \; \mathcal{N}_{4r}${ \normalfont is a region of expansion}}\right\},
\end{align*}
then $D$ occurs w.h.p.
\end{lemma*}

\begin{proof}
Since $D$ is increasing in a flip of a (-1) agent, we can focus on the case when the initial configuration is preserved. 
In this case, for the configuration to be expandable we need to make sure that any agent right outside the boundary of a monochromatic $w$-block will be unhappy. We  obtain a lower bound for the probability of this event. With the same argument as in the proof of Lemma \ref{Lemma:unhappyprob} of the Appendix, a lower bound for the probability that a given agent right outside the boundary of a monochromatic neighborhood $\mathcal{N}_{w/2}$ is unhappy, is 
\begin{align*}
1-2^{-[1-H(\frac{4}{3}\tau)]\frac{3}{4}{\size}-o({\size})}.
\end{align*}
Let us denote the latter event for the (-1) agents right outside the boundary of $\mathcal{N}_{w/2}$ by $A_1, ... , A_L$, where $L$ is the number of (-1) agents right outside the boundary of  $\mathcal{N}_{w/2}$. It is easy to see that these   are all increasing events and using  the FKG inequality we conclude that
\begin{align*}
P(A_1\cap ... \cap A_L) \ge P(A_1)P(A_2)...P(A_L) \ge (1-2^{-[1-H(\frac{4}{3}\tau)]\frac{3}{4}{\size}-o({\size})})^L.
\end{align*}
Now, for any $v \in \mathcal{N}_{4r}$  let $B_v$  be the event that all the (-1) agents outside $\mathcal{N}_{w/2}$ centered at $v$ are unhappy. It is also easy to see that $B_v$'s are increasing events. Hence, with another application of the FKG inequality we have
\begin{align*}
P\left(\bigcap\limits_{v \in \mathcal{N}_{4r}} B_v\right) \ge \left(1-2^{-[1-H(\frac{4}{3}\tau)]\frac{3}{4}{\size}-o({\size})}\right)^{2^{[1-H(\tau)]{\size}+o({\size})}},
\end{align*}
where we have used the fact that $L < \size$.  
\end{proof}
 



 Consider a disc of radius $r$,   centered at an agent such that all the agents inside the disc are of the same type. It is easy to see that if $r$ is sufficiently large then all the agents inside the disc will remain happy regardless of the configuration of the agents outside the disc. Lemma 6 in \cite{immorlica2015exponential} shows that for $r>w^3$ this would be the case for sufficiently large $w$. Here we state a similar lemma but for an annulus, i.e., a firewall, without proof.
 
\begin{lemma*}   \label{Lemma:firewall}
Let $A_r(u)$ be the set of agents contained in an annulus of outer radius $r \ge w^3$ and of width $\sqrt{2} w$ centered at $u$. For all $\tau \in (\tau_2,1/2)$ and for a sufficiently large constant $w$, if $A_r(u)$ is monochromatic at time $t$,  then it will remain monochromatic at all times  $t'>t$. 
\end{lemma*}

\begin{lemma*}\label{Lemma:fw}
Let $\mathcal{N}_{\radius}$, $\mathcal{N}_{{\radius}/2}$, $\mathcal{N}_{4r}$, and $\mathcal{N}_{r}$ be  all centered at  $u$  with  ${\radius} = 2^{[1-H(\tau')]{\size}/2}$ and $r = 2^{[1-H(\tau')]{\size}/2-o({\size})}$, $r<{\radius}/8$. Let $u^+$ denote an arbitrary (+1) agent, $T({\radius})$ be as defined in (\ref{Tinfdef}), and $\kappa$ be such that  $\kappa  r {\size}^{1/2}$ is the sum of the number of agents in a firewall with radius $2r$ and the number of agents in a line of width $w+1$ that connects the center to the boundary of the firewall and includes   $\mathcal{N}_{w/2}$ at its center. Conditioned on the following events, w.h.p. the monochromatic region of $u$ will have at least radius $r$.
\begin{enumerate} 
\item $A = \left\{\forall v\in \mathcal{N}_{\radius}, \;  u^{+}   \mbox{ {\normalfont would be happy at the location of $v$ at $t=0$}}\right\}, $
\item $B = \{T({\radius}/2) > 2\kappa  r {\size}^{1/2} \}$,
\item $C = \left\{\text{$\mathcal{N}_r$  {\normalfont contains an expandable radical region at }$t=0$}\right\}$,
\item $D = \text{\{$\forall t < T({\radius}/2), \;\; \mathcal{N}_{4r}$  {\normalfont is a region of expansion}\}}$.
\end{enumerate}
\end{lemma*}

\begin{figure}[!t]
\centering
\includegraphics[width=\textwidth]{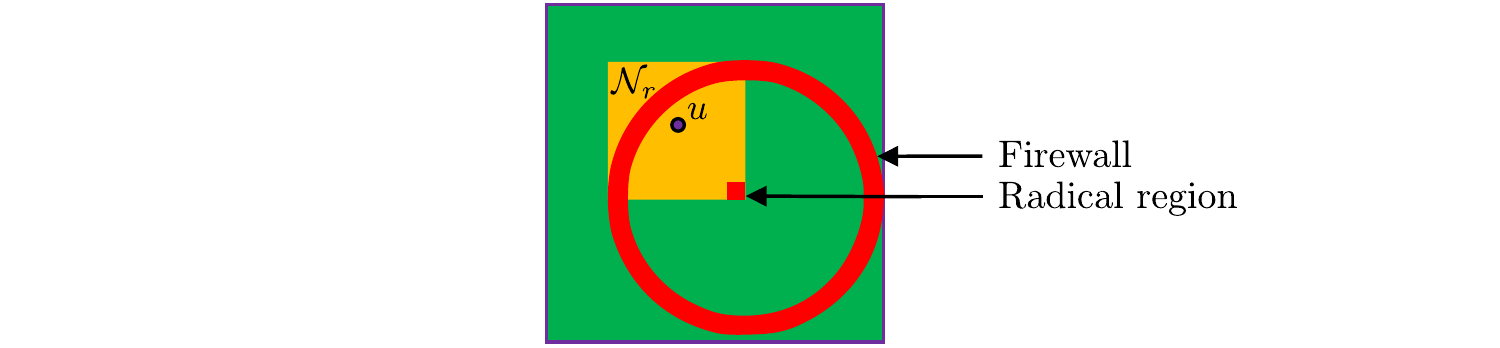}
\caption{Neighborhoods described in the proof of Lemma~\ref{Lemma:fw}.}
\label{fig:thrm_fw}
\end{figure}

\begin{proof}

Conditioned on events $A, B$, $C$, and $D$,  an expandable radical region contained in $\mathcal{N}_r$ can lead to the formation of a firewall of radius $2r$ centered at this region. Let $M(r)$ denote the event that the radius of the monochromatic region of $u$ is at least $r$. Let $T_f$ be the time at which this firewall forms, meaning that all the agents contained in the annulus become of the same type. We have
\begin{align*}
P\left(M(r)\given[\Big]A, B , C , D\right) \ge P\left(T_f<2\kappa  r \sqrt{{\size}}  \given[\Big]A,B,C,D\right) 
\end{align*}
Let $T'_f$ be the sum of $\kappa  r {\size}^{1/2}$ exponential random variables with mean one. It is easy to see that $T'_f$ is an upper bound for the time it takes until the firewall is formed, since the worst case scenario for the formation of the firewall is when the $\kappa  r {\size}^{1/2}$ agents flip to (+1), one by one. Hence, we have
\begin{align*}
 P\left(T_f<2\kappa  r \sqrt{{\size}}  \given[\Big]A,B,C,D\right)  \ge P\left(T'_f<2\kappa  r \sqrt{{\size}}\right). 
\end{align*}
Next, we bound this probability. We have
\begin{align*}
P\left(T'_f\ge 2\kappa  r \sqrt{{\size}}\right) &\le P\left(|T'_f-\mathbb{E}[T'_f]|\ge \kappa r\sqrt{{\size}}\right).
\end{align*}
By Chebyshev's inequality, we  have
\begin{align*}
P\left(T'_f\ge 2\kappa  r \sqrt{{\size}}\right) = O\left(\frac{\mbox{Var }(T'_f)}{(r\sqrt{{\size}})^2}\right) = O\left(\frac{r\sqrt{{\size}}}{(r\sqrt{{\size}})^2}\right) = O\left(\frac{1}{r\sqrt{{\size}}}\right).
\end{align*}
It follows that w.h.p. agent $u$ will be trapped inside a firewall together with an expandable radical region and the interior of the firewall will be a region of expansion until the end of the process. Hence this interior will eventually become monochromatic and, as a result, agent $u$ will have a monochromatic region of size at least proportional to $r^2$, as desired.  
\end{proof} 

We can now  give the proof for the first part of Theorem~\ref{Thrm:first_theorem}. 

\noindent {\bf Proof of Theorem~\ref{Thrm:first_theorem}} {{\normalfont (for $\tau_1< \tau< 1/2$)}}
First, we derive the lower bound in the theorem letting 
\begin{equation}
a(\tau) = \left[1-(2\epsilon'+\epsilon'^2)\right]\left[1-H(\tau')\right],
\end{equation}
where $\epsilon' > f(\tau)$, and $\tau'=(\tau {\size} -2)/({\size}-1)$.

We consider  neighborhoods $\mathcal{N}_{\radius}$,  $\mathcal{N}_{\radius/2}$,  and  $\mathcal{N}_r$,  with ${\radius}  =  2^{[1-H(\tau')]{\size}/2}$ and $r< \radius/8$,
all centered at node $u$   as depicted in Figure~\ref{fig:thrm1}.
We let
$u^{+}$ be an arbitrary (+1) agent, and consider the following event in the initial configuration
\begin{align}
A = \left\{\forall v\in \mathcal{N}_{\radius}, \;  u^{+}   \mbox{ {\normalfont would be happy at the location of $v$ at  $t=0$}}\right\}.
\end{align}
By Lemma \ref{Lemma:R_unhappy} of the Appendix, we have
\begin{align}
P(A) \rightarrow 1,\  \mbox{ as } \    {\size} \rightarrow \infty.
\label{pa}
\end{align}
We then consider a firewall of radius $2r$ centered anywhere inside $\mathcal{N}_r$, let $\kappa >0$ so that $\kappa  r {\size}^{1/2}$ is the sum of the number of agents in it and the number of agents in a line of width $w+1$ that connects its center to its boundary and includes $\mathcal{N}_{w/2}$ at its center. Consider the event
\begin{align*} 
B = \left\{\text{$T({\radius}/2) > 2\kappa  r {\size}^{1/2}$}\right\},
\end{align*}
where $T({\radius})$ is defined in (\ref{Tinfdef}).
By  Lemma~\ref{Lemma:Unhappy_growth}, we can choose $r$ proportional to  ${\radius}/({\size}^2)$ so  that
\begin{align}
P(B|A) \rightarrow 1,\  \mbox{ as } \    {\size} \rightarrow \infty.
\label{pea}
\end{align}
With this choice, we also have
\begin{align*}
r &= 2^{[1-H(\tau')]{\size}/2-o({\size})},
\end{align*}
and if we consider the event
\begin{align*}
C = \left\{\text{$\mathcal{N}_r$  contains an expandable radical region at $t=0$}\right\},
\end{align*}
by Lemma \ref{Lemma:expandable}, we have  for $N$ sufficiently large
\begin{align}
P(C) \geq 2^{-[1-H(\tau')](2\epsilon'+\epsilon'^2){\size} -o({\size})}.
\label{pc}
\end{align}
Consider a neighborhood $\mathcal{N}_{4r}$ also centered at $u$ and the event
\begin{align*} 
D = \left\{\text{$\forall t < T({\radius}/2), \; \mathcal{N}_{4r}$  is a region of expansion}\right\}.
\end{align*}
By Lemma \ref{Lemma:monoch_spread_1}, we have
\begin{align}
P(D) \rightarrow 1,\  \mbox{ as } \    {\size} \rightarrow \infty.
\label{pd}
\end{align}

 \begin{figure}[!t]
\centering
\includegraphics[width=\textwidth]{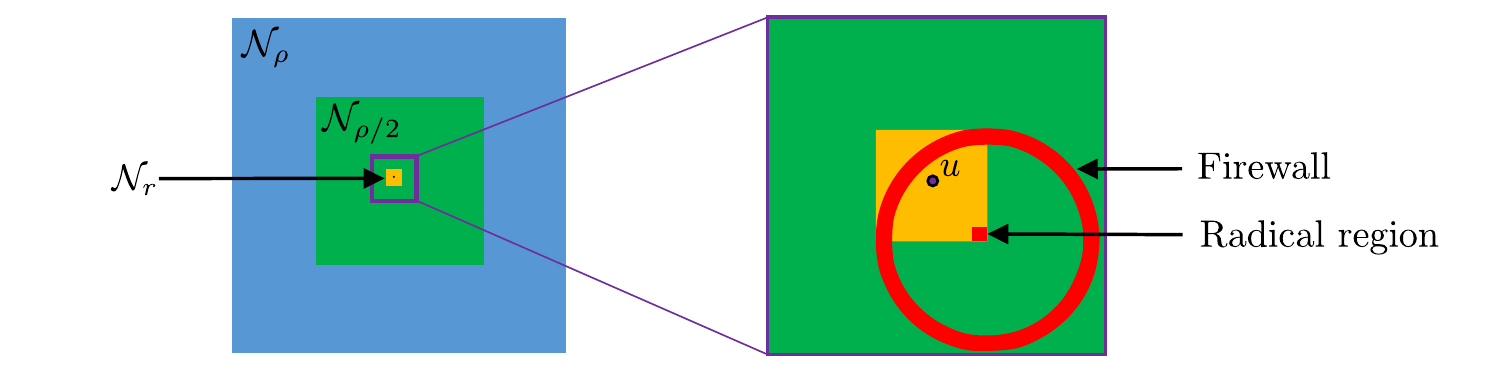}
\caption{Neighborhoods described in the proof of Theorem~\ref{Thrm:first_theorem}.}
\label{fig:thrm1}
\end{figure}

We now note that  $A$, $B$, $C$, $D$ are increasing events with respect to a partial ordering on their outcomes.
More precisely, consider two outcomes of the sample space $\omega, \omega' \in \Omega$ such that $\omega, \omega' \in E$ where $E$ is an event. We define a partial ordering on the outcomes such that $\omega' \ge \omega$ if for all time steps, the set of agents of type (+1) in $\omega$ is a subset of the set of agents of type (+1) in $\omega'$. Event $E$ is increasing if $1_E(\omega') \ge 1_E(\omega)$ where $1_E$ is the indicator function of the event $E$. According to this definition, $A$, $B$, $C$,   $D$ are increasing events.
By combining (\ref{pa}), (\ref{pea}), (\ref{pc}), and (\ref{pd}), and using  a version of the FKG inequality adapted  to our dynamic process, stated in Lemma~\ref{Lemma:FKG_Harris} of the Appendix, it follows  that for ${\size}$ sufficiently large 
\begin{align}
P(A \cap B \cap C \cap D) &\ge P(A)P(B)P(C)P(D) \nonumber \\
&\ge P(A)P(B\cap A)P(C)P(D) \nonumber \\
&=  P(B|A)[P(A)]^2P(C)P(D) \nonumber \\
&=  2^{-[1-H(\tau')][2\epsilon'+\epsilon'^2]{\size} -o({\size})}.\label{Eq:intersect1}
\end{align} 
Since by Lemma~\ref{Lemma:fw} we have that conditioning on  $A, B, C,$ and $D$, at the end of the process w.h.p. agent $u$ will be  part of a monochromatic region with radius at least $r$, it follows that (\ref{Eq:intersect1}) is also a lower bound for the probability that the monochromatic neighborhood of agent $u$ will have size of at least proportional to $r^2$. The desired lower bound on the expected size of the monochromatic region now easily follows by multiplying  (\ref{Eq:intersect1}) by the size  of a neighborhood of radius $r$.

 
Next, we show the corresponding upper bound, letting
\begin{align*} 
b(\tau)= \left[\frac{3}{2}(1+\epsilon')^2\right][1-H(\tau')], 
\end{align*}
and $\epsilon'$ and $\tau'$  as defined above. For any $\delta>0$, consider a neighborhood $\mathcal{N}_{{\radius}'}$ such that 
\begin{align*}
{\radius}' = 2^{(1+\epsilon')^2[1-H(\tau')]{\size}/2+ \delta {\size}/2},
\end{align*}
 and divide $\mathcal{N}_{{\radius}'}$ into blocks of size $\mathcal{N}_{\radius}$ in the obvious way. Let $M_{+1}$ and $M_{-1}$ denote the events of   $\mathcal{N}_{{\radius}'}$ being monochromatic of type (+1) and (-1) respectively. Also let $E_{+1}$ and $E_{-1}$ be the events of having a monochromatic region of type (+1) and  (-1) inside a firewall of radius $2r$ centered anywhere  inside $\mathcal{N}_{{\radius}'}$. 
We have  that for $N$ sufficiently large 
\begin{align}
P(M_{+1}\cup M_{-1}) &\le P(M_{+1})+P(M_{-1}) \nonumber \\
&=P(M_{+1}\cap E^C_{-1}) +  P(M_{-1}\cap E^C_{+1}) \nonumber \\
& \le P(E^C_{-1}) + P(E^C_{+1}) \nonumber \\
& = 2P(E^C_{-1}) \nonumber \\
& \le 2(1-2^{-[1-H(\tau')]\left(2\epsilon'+\epsilon'^2 \right){\size}-o({\size})})^{{\radius}'^2/{\radius}^2} \nonumber \\
  &= e^{-2^{\delta {\size} - o({\size})}}.
  \label{finalM}
\end{align}

By considering the set of all the neighborhoods  of radius $\radius'$ sharing agent $u$, by the union bound the probability that at least one of them will be monochromatic of only one type is also bounded by (\ref{finalM}).
We now consider the expected size of the monochromatic region of agent $u$, that is bounded as
\begin{align*}
\mathbb{E}[M] \le \sum_{m=1}^{n}m^2p_m,
\end{align*}
where $p_m$ denotes the probability of having a monochromatic region of size $m^2$ containing  $u$. 
We let
\begin{align*}
{\radius}'' = 2^{[(1+\epsilon')^2(1-H(\tau')]{\size}/2+ o({\size})},
\end{align*} and divide the series into two parts
\begin{align}\label{Eq:upper_bound}
\mathbb{E}[M] &\le \sum_{m=1}^{{\radius}''}m^2p_m + \sum_{m={\radius}''+1}^{n}m^2p_m \nonumber \\
&\le 2^{\left[\frac{3}{2}(1+\epsilon')^2(1-H(\tau')\right]{\size}+ o({\size})} + \sum_{m={\radius}''+1}^{n}m^2p_m,
\end{align}
where the first inequality follows from $p_m \leq 1$.
Since by (\ref{finalM})  for all $m \geq {\radius}'$, the probability of having a monochromatic region of  size $m^2$ containing  $u$ has at most a double exponentially small probability, the tail of the remaining series in (\ref{Eq:upper_bound}) converges to a constant, while    for sufficiently large $N$ the sum of the first $\radius'-\radius''-1$ terms is smaller than the first term of  (\ref{Eq:upper_bound}), and the proof is complete.  

\subsection{Almost monochromatic region}
We now turn our attention to the case where $\tau \in (\tau_2,\tau_1]$. 
We define an $\mbox{$m$-block}$ to be a neighborhood of radius $m/2$. Let $\mathcal{I}$ be the collection of sets of agents in the possible intersections of a $w$-block with an $m$-block on the grid in the initial configuration. Also, let $W_I$ be the random variable representing the number of (-1)'s in $I \in \mathcal{I}$, and ${\size}_I$ be the total number of agents in $I \in \mathcal{I}$.


\

 \begin{figure}[!t]
\centering
\includegraphics[width=2.5in]{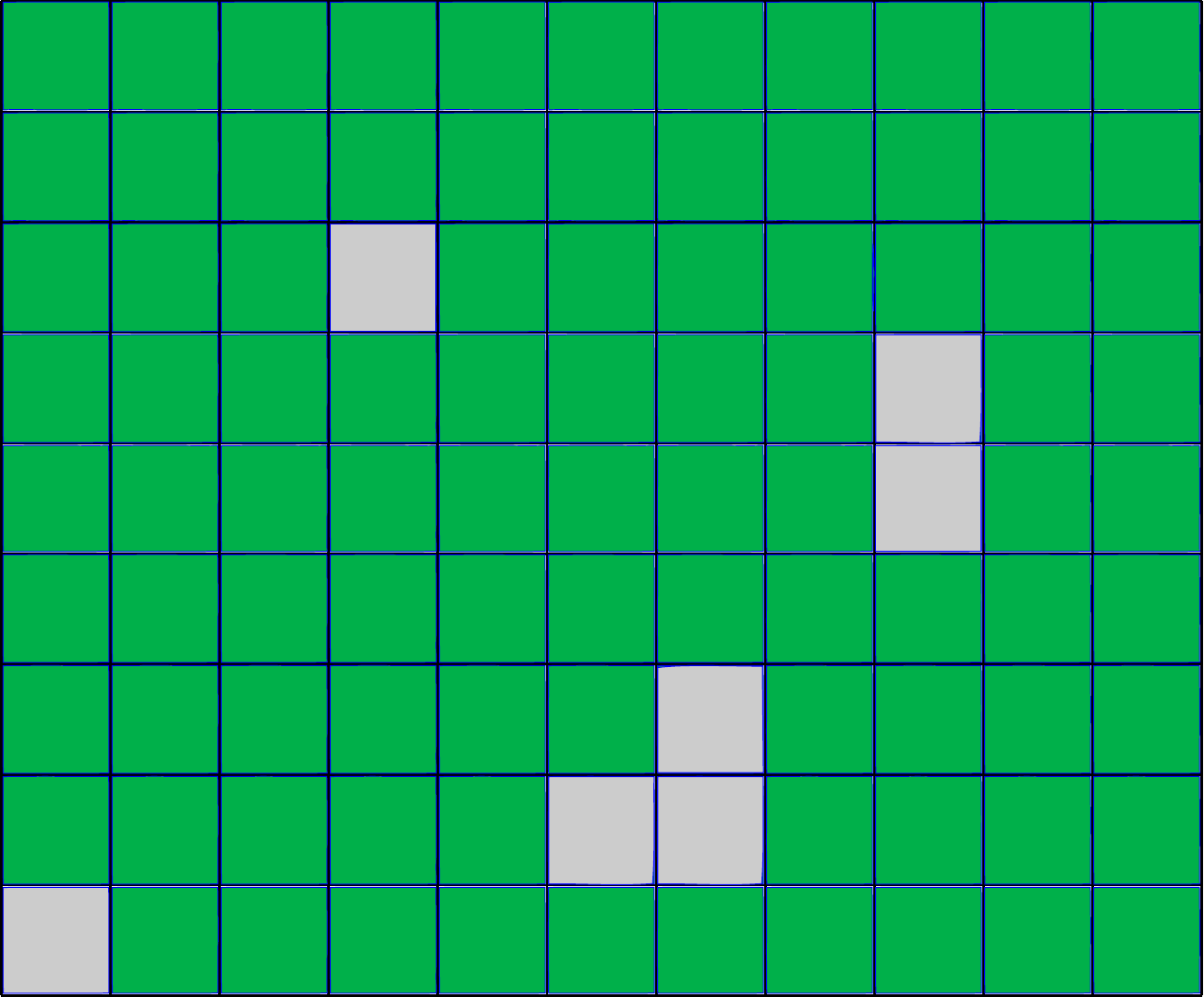}
\caption{Part of the grid renormalized into $m$-blocks.  Green and gray indicate good and bad blocks respectively. }
\label{fig:goodbad}
\end{figure}
\textit{Good block.} For any $\epsilon \in (0,1/2)$, a \textit{good $m$-block} is an $m$-block such that for all $I\in \mathcal{I}$ we have $W_I-{\size}_I/2  < {\size}^{1/2+\epsilon}$. The $m$-blocks that do not satisfy this property are called \textit{bad $m$-blocks} (see Fig. \ref{fig:goodbad}). It is easy to see that all the blocks contained in a good $m$-block are also good blocks.

\

For the following two definitions, we assume that the grid is renormalized into $m$-blocks. In this setting each $m$-block is horizontally or vertically adjacent to four other $m$-blocks.

\textit{$m$-path}. An $m$-path is an ordered set of $m$-blocks such that each pair of consecutive $m$-blocks are either horizontally or vertically adjacent and no $m$-block appears more than once in the set. The \textit{length} of the path is the number of $m$-blocks in the path. Two $m$-blocks are \textit{connected} if there exists an $m$-path between them.

\

\textit{$m$-cycle}.
An $m$-cycle is a closed path in which the last $m$-block in its ordered set is adjacent to the first $m$-block. An $m$-cycle divides the  $m$-blocks of the grid into two sets of $m$-blocks referred to as its \textit{interior} and its \textit{exterior}.

\

\textit{$r$-chemical path}.
Renormalize the grid into $\bsize$-blocks starting from the block centered at agent $u$.
To define an $r$-chemical path, consider two neighborhoods $\mathcal{N}_{3r}$ and $\mathcal{N}_{r}$ with radii $3r$ and $r$ respectively and both centered at an agent $u$.  

Let $r>12w^3$. An \textit{$r$-chemical path} centered at $u$, is the union of a $\bsize$-cycle of good $\bsize$-blocks contained in $\mathcal{N}_{3r} \setminus \mathcal{N}_{r}$ such that $u$ is in its interior, and a path of good $\bsize$-blocks from the $\bsize$-block at the center of $\mathcal{N}_r$ to a $\bsize$-block in the $\bsize$-cycle, such that the total length of the $\bsize$-cycle and the $\bsize$-path is proportional to $r/(\bsize)$  (see Fig. \ref{fig:chemical_in_r}).

\

\begin{figure}[!t]
\centering
\includegraphics[width=2.5in]{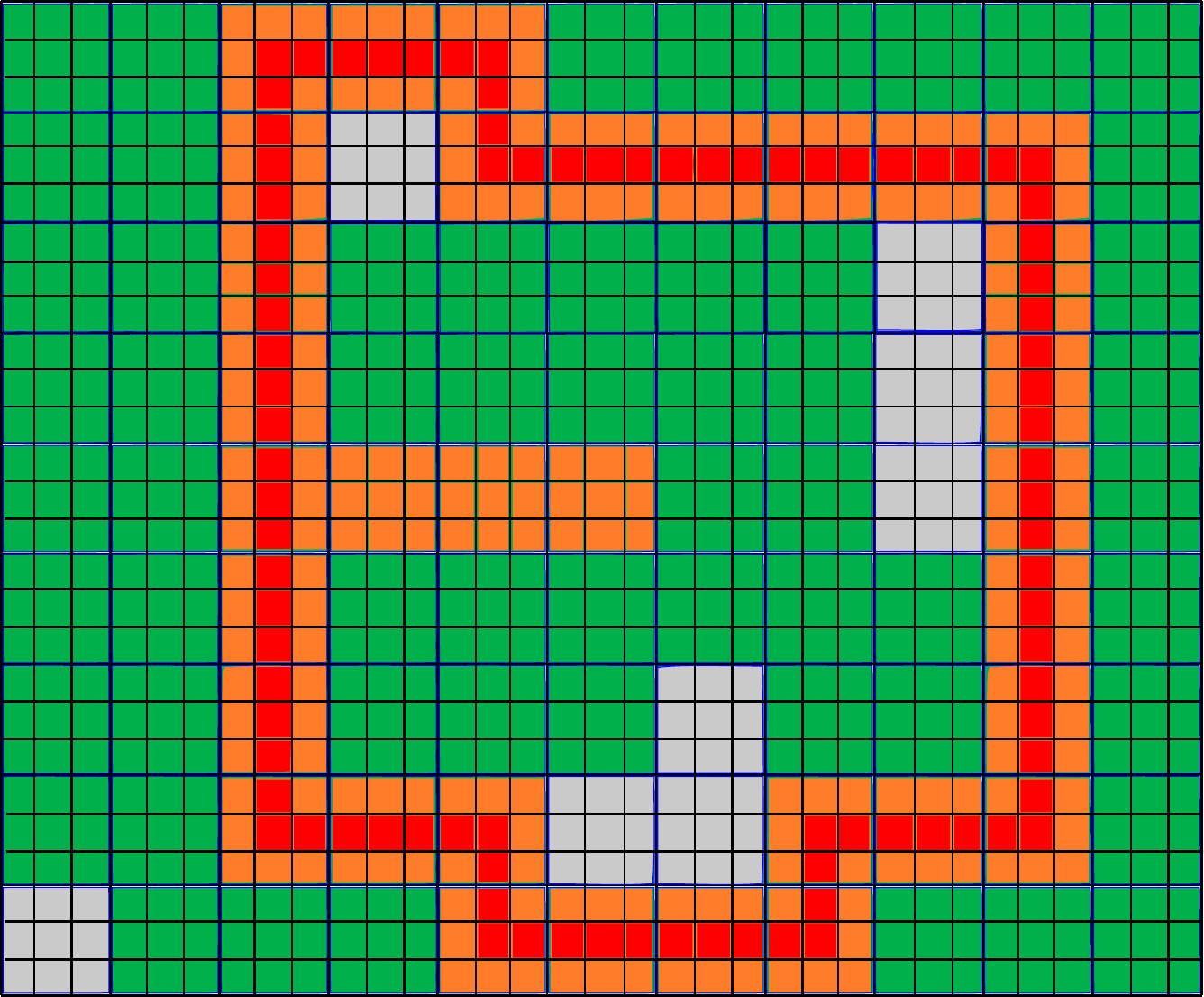}
\caption{  Larger blocks are $\bsize$-blocks and smaller ones are $\bfsize$-blocks. The red cycle indicates the chemical firewall which is in the cycle of an $r$-chemical path (orange). }
\label{fig:chemical_in_r}
\end{figure}

\textit{Chemical firewall.}  
Renormalize the grid into $\bfsize$-blocks starting from the block centered at agent $u$ and consider the $r$-chemical path defined above in this setting.
A \textit{chemical firewall} with radius $r$ is a $\bfsize$-cycle contained in the cycle of the $r$-chemical path such that agent $u$ is in its interior and all the agents in the $\bfsize$-cycle are of the same type (see Fig. \ref{fig:chemical_in_r}).


Although the structure of a chemical firewall is very different from the annular firewall defined before, the size of the $m$-blocks are chosen such that it is easy to see that, with similar arguments given for Lemma \ref{Lemma:firewall}, it acts as a firewall, i.e., the flips of the agents in its exterior cannot affect the agents in its interior.

\

An \textit{$r$-expandable radical region} of type (-1) is a radical region such that it is  expandable and it is located at the center of an $r$-chemical path.

Before proceeding with the first part of the proof of Theorem~\ref{Thrm:second_theorem}, we need the following results. The following lemma gives a lower bound for the probability that an arbitrary $m$-block with $m \le \size^3$ is a good $m$-block. Using this lemma, by renormalizing the grid into $m$-blocks we will argue that the probability that a block is a bad block can be arbitrary small for sufficiently large $\size$. 

\begin{lemma*} \label{Lemma:goodblock}
Let $\epsilon \in (0,1/2)$ and $m\le \size^3$. For all $I\in \mathcal{I}$ we have $W_I-{\size}_I/2 < {\size}^{1/2+\epsilon}$ with probability at least
\begin{align*}
1 - e^{-c{\size}^{2\epsilon}+o({\size}^{2\epsilon})}.
\end{align*}
\end{lemma*}

 \begin{proof}
By Lemma \ref{Lemma:balanced} of the Appendix, for an arbitrary $I\in \mathcal{I}$  we have
\begin{align*}
P\left(W_I-{\size}_I/2\ge {\size}^{1/2+\epsilon}\right) < e^{-c{\size}^{2\epsilon}},
\end{align*}
where $\epsilon \in (0,1/2)$ and $c>0$.  Since there are less than ${\size}^3$ elements in $\mathcal{I}$, we   have
\begin{align*}
P\left(W_I-{\size}_I/2< {\size}_I^{1/2+\epsilon} \mbox{ for all }I\in \mathcal{I}\right) \ge 1 - {\size}^{3}e^{-c{\size}^{2\epsilon}}.
\end{align*}  
\end{proof}

Let us consider a neighborhood consisting of exponentially large number of $m$-blocks where $m\le \size^3$. Based on the following lemma,  the ratio between bad blocks and good blocks in this neighborhood is exponentially small w.h.p. 
 
\begin{lemma*} \label{Lemma:bad_ratio}
Let $c$ be a positive constant and $\epsilon \in (0,1/2)$. Let $\mathcal{N}_{\radius}$ be a neighborhood consisting of  $m$-blocks and with $2^{c{\size}}$ agents. The ratio between bad blocks and good blocks is less than $e^{-{\size}^\epsilon}$ w.h.p.
\end{lemma*}

\begin{proof}
By Lemma \ref{Lemma:goodblock}, the probability of having a bad block is less than $e^{-{\size}^{2\epsilon}+o({\size}^{2\epsilon})}$. It is easy to show that the number of bad blocks is less than $2^{c{\size}}e^{-{\size}^{2\epsilon}+o({\size}^{2\epsilon})}$ w.h.p. Hence, the ratio between the number of bad blocks and the number of good blocks is less than $e^{-{\size}^\epsilon}$ w.h.p., see Figure \ref{fig:goodbad}.  
\end{proof}

We now want to argue that the formation of a chemical firewall is likely. 
We first notice that a monochromatic $w$-block located inside a good $\bsize$-block can make at least a $\bfsize$-block at the center of the good block monochromatic. This means that a monochromatic $w$-block at the center of the $r$-chemical path can create a chemical firewall (see Fig. \ref{fig:chemical_in_r}). Our next goal is to show that the existence of an $r$-chemical path is likely.  The critical step is to show that the length of the $r$-chemical path  is proportional to $r/\bsize$. 

We use a result from percolation theory \cite{garet2007large} restated in the following. Consider site percolation on square lattice in the supercritical regime. Let $D(0,x) = \inf_{\Gamma} |\Gamma|$, where $\Gamma$ is a path from the origin to the vertex $x$ and $|\Gamma|$ is the number of vertices in the path. Let $0\leftrightarrow x$ denote that $0$ and $x$ belong to the same connected component. The following   is Theorem~1.4 from \cite{garet2007large}, and it asserts that the length of the shortest path between the origin and an arbitrary vertex $x$ cannot be much different from its $l_1$   distance $\| x \|_1$, see Figure~\ref{fig:chemical}.

\begin{figure}[!t]
\centering
\includegraphics[width=2.5in]{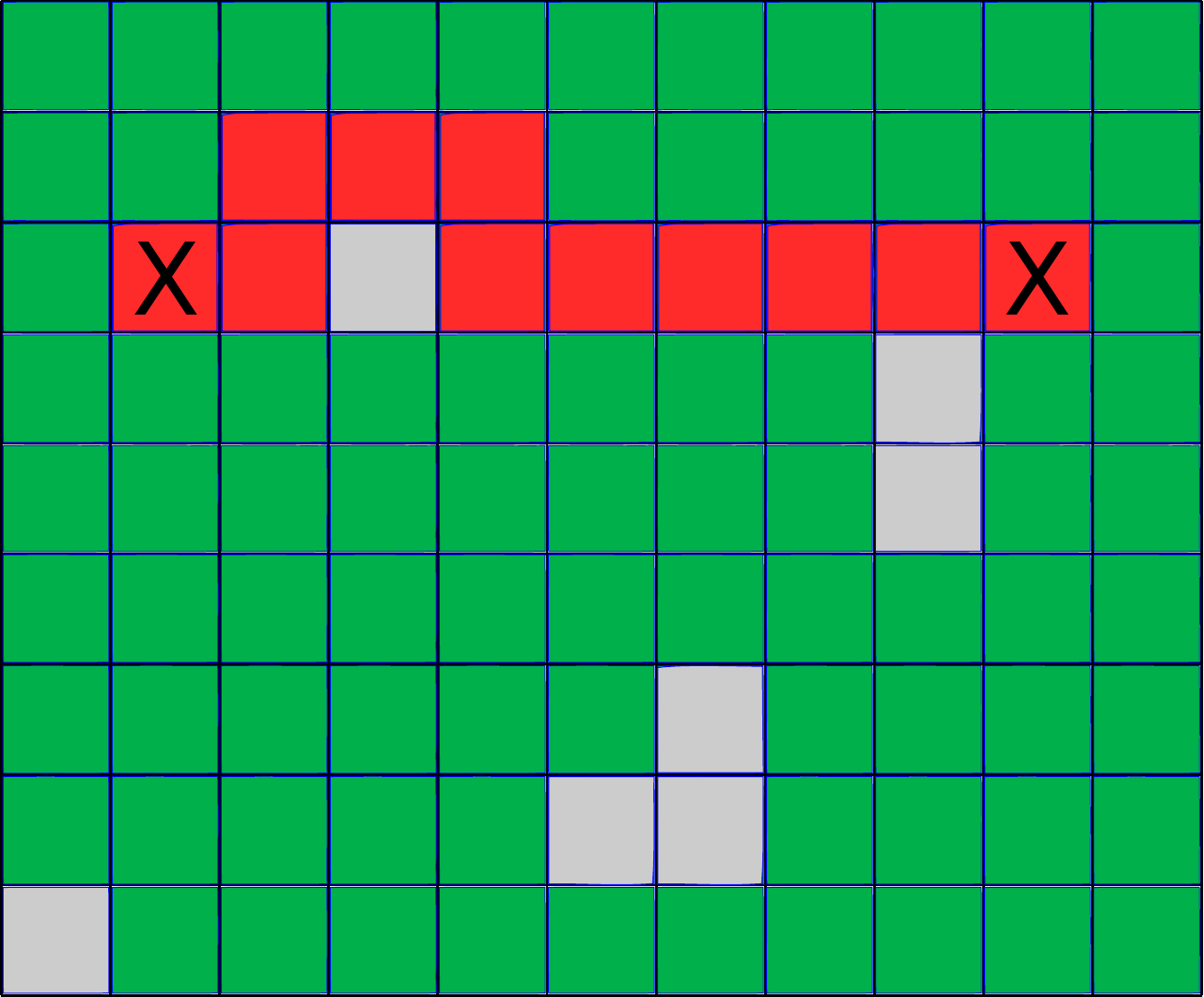}
\caption{ The length of the shortest path of good blocks between two arbitrary vertices denoted by X is w.h.p. not much different from its $l_1$-distance between them in the supercritical regime. }
\label{fig:chemical}
\end{figure}

\begin{theorem*}[Garet and Marchand] \label{Thrm:Garet}
For all $\alpha > 0$, there exists $p'(\alpha) \in (p_c(d),1)$ such that for all $p \in (p'(\alpha),1]$, we have:
\begin{align*}
\limsup_{\|x\|_1\rightarrow +\infty} \frac{\ln P_p\left(0\leftrightarrow x,D(0,x) \ge (1+\alpha)\|x\|_1\right)}{\|x\|_1} < 0.
\end{align*}

\end{theorem*} 
Now consider a two dimensional lattice which consists of good $\bsize$-blocks and bad $\bsize$-blocks. The probability of a site being good then, is at least the value computed in Lemma \ref{Lemma:goodblock}, hence for sufficiently large $\size$ we are dealing with a percolation problem in the super-critical regime. Let us denote a radical region with radius $\epsilon'$ by $\epsilon'$-\textit{radical region}.

\begin{lemma*}
 \label{Lemma:firewall_path}
W.h.p. an $\epsilon'$-radical region is at the center of an $r$-chemical path at time $t=0$ where $r < n/10$.
\end{lemma*}

\begin{proof}
Since an $r$-chemical path is contained in a neighborhood of radius $3r$, without loss of generality we can assume that this neighborhood is contained in a $\mathbb{Z}^2$ lattice. It is also clear that the flip of a (-1) agent, can only increase the probability of formation of the $r$-chemical path. Divide the resulting lattice into $m$-blocks such that the $\epsilon'$-radical region is at the center of an $m$-block and call the resulting renormalized  lattice $\mathbb{L'}$. Consider performing site percolation on this lattice by considering good $\bsize$-blocks as open sites of $\mathbb{L'}$ and bad $\bsize$-blocks as its closed sites. As discussed above, for sufficiently large $\size$ we are dealing with a percolation problem in its super-critical regime. Consider two blocks containing agents $(2r,2r)$ and $(-2r,2r)$ in the original lattice denoted by $0$ and $x$ respectively. By Theorem~\ref{Thrm:Garet} we conclude that for sufficiently large ${\size}$ there exists a constant $c>0$ such that
\begin{align*}
P_p\left(0\leftrightarrow x,D(0,x) \ge (1.25)\| x \|_1\right) \le e^{-c\|x\|_1}
\end{align*}
 where $\|x\|_1$ is the $l_1$ distance of $x$ from $0$ and we have put $\alpha = 0.25$. By the union bound and the FKG inequality, we have 
 \begin{align*}
 P_p\left(D(0,x) < 1.25\| x \|_1\right) &\ge P\left(0\leftrightarrow x\right) - P\left(0\leftrightarrow x,D(0,x) \ge (1.25) \| x \|_1\right) \\
  &\ge \theta(p)^2 - e^{-c \| x \|_1},
 \end{align*} 
where $\theta(p)$ is the probability that a node belongs to an infinite cluster and we have used the FKG inequality to conclude that $P(0\leftrightarrow x)\ge \theta(p)^2$. Now, using Lemma \ref{Lemma:goodblock} it is easy to see that for sufficiently large values of ${\size}$ this lower bound is as close as we want to one.

For each pair of corner agents of $\mathcal{N}_{2r}$ on the same side the above argument holds. A similar argument also holds for the existence of a path from the center of $\mathcal{N}_r$ to an arbitrary block on the boundary of $\mathcal{N}_{3r}$, i.e., a $\bsize$-block which contains agents with $l_\infty$-distance of $3r$ from the center of $\mathcal{N}_{3r}$. It is also easy to see that these events are all increasing events, i.e., their indicator functions can only increase by changing a closed site to an open site, in this case, a bad $\bsize$-block to a good $\bsize$-block. Hence, by the FKG inequality,   the joint probability of the existence of the above paths is at least their product which can be made arbitrary close to one for large values of ${\size}$.  
\end{proof}

We need to show that w.h.p. the radical region located inside the firewall can make the interior of the firewall almost monochromatic by the end of the process. We show that there are no clusters of bad blocks of radius larger than a polynomial function of ${\size}$ in a neighborhood with exponential size in $\size$.
To show this we first restate a result from \cite{grimmett1999percolation}. Let $S(k)$ be the ball of radius $k$ with center at the origin, i.e., $S(k)$ is the set of all vertices $x$ in $\mathbb{Z}^2$ for which $\Delta(0,x) \le k$, where $\Delta$ denotes the $l_1$ distance. Let $\partial S(k)$ denote the surface of $S(k)$, i.e., the set of all $x$ such that $\Delta(0,x)=k$. Let $A_k$ be the event that there exists an open path joining the origin to some vertex in $\partial S(k)$. Let the \textit{radius} of a bad cluster be defined as 
\begin{align*}
\sup \{\Delta(0,x): x \in \mbox{bad cluster}\}.
\end{align*}
 The following result is Theorem~5.4 in \cite{grimmett1999percolation}.

\begin{theorem*} [Grimmett] \label{Thrm:grimmett_bad_cluster}
 \textsl{(Exponential tail decay of the radius of an open cluster.)}
If $p<p_c$, there exists $\psi(p) >0 $ such that
\begin{align*}
P_p(A_k) < e^{-k\psi(p)}, \ \ \  for \ all \ k.
\end{align*}
\end{theorem*}


\begin{lemma*}\label{Lemma:bad_cluster}
W.h.p. there are no clusters of bad $\bsize$-blocks with radius greater than ${\size}^2$ blocks in a neighborhood with radius $4r= 2^{[1-H(\tau')]{\size}/2-o({\size})}$ at time $t=0$.
\end{lemma*}
\begin{proof}
Let $p$, 
be the probability of having a bad $\bsize$-block, and let $k={\size}^2$. By Theorem~\ref{Thrm:grimmett_bad_cluster} it follows that w.h.p. there is no cluster of bad $\bsize$-blocks containing a bad $\bsize$-block with $l_1$-distance from its center greater than ${\size}^2$ $\bsize$-blocks in a neighborhood with exponential radius in ${\size}$.  
\end{proof}


\

It is easy to check that for $\tau > 3/8$, a monochromatic $w$-block in a good block can   make the whole block monochromatic (except for possibly  a margin of $w$ at the borders). On the other hand, Lemma~\ref{Lemma:Sufficient2} shows that the same condition of Lemma~\ref{Lemma:Trigger}  leads to the formation of a monochromatic $3w/2$-block for $\tau \in (\tau_1,3/8)$ because once the $\epsilon'$-radical region leads to a monochromatic $w$-block at its center, it can as well lead w.h.p  to a monochromatic $3w/2$-block. 
Lemma \ref{Lemma:Monoch_spread2}  then shows that the spread of the monochromatic $3w/2$-blocks is indeed possible.

\begin{lemma*} \label{Lemma:Sufficient2}
Consider the $\mathcal{N}_S$ neighborhood defined in Lemma \ref{Lemma:Trigger} and co-centered with a neighborhood $\mathcal{N}_{\radius}$ of radius ${\radius}>{\size}$ with the property that no (+1) agent inside $\mathcal{N}_{\radius}$ will become unhappy until some time $T(\radius)$. Then w.h.p.   there exists a set of flips with the following property: if they  happen before $T(\radius)$ then  all the agents inside a neighborhood with radius $3w/2$   concentric with $\mathcal{N}_{\radius}$   will be of the same type.
\end{lemma*}

\begin{proof}
By  Lemma \ref{Lemma:Trigger}, w.h.p. there exists a set of flips that if they happen before $T(\radius)$ will make a $w$-block at the center of $\mathcal{N}_{\radius}$ unhappy. By Proposition \ref{Prop:firstprop}, it follows that this monochromatic block will make all the (-1) agents in four identical trapezoids outside the $w$-block whose larger bases are the sides of the $w$-block unhappy, and hence monochromatic w.h.p. Now, with another application of Proposition \ref{Prop:firstprop} we have  that for $\tau > \tau_1$, all the (-1) agents in a $3w/2$-block with the same center as the $w$-block will be unhappy, hence the $3w/2$-block can become monochromatic w.h.p.   
\end{proof}

\begin{figure}[!t]
\centering
\includegraphics[width=\textwidth]{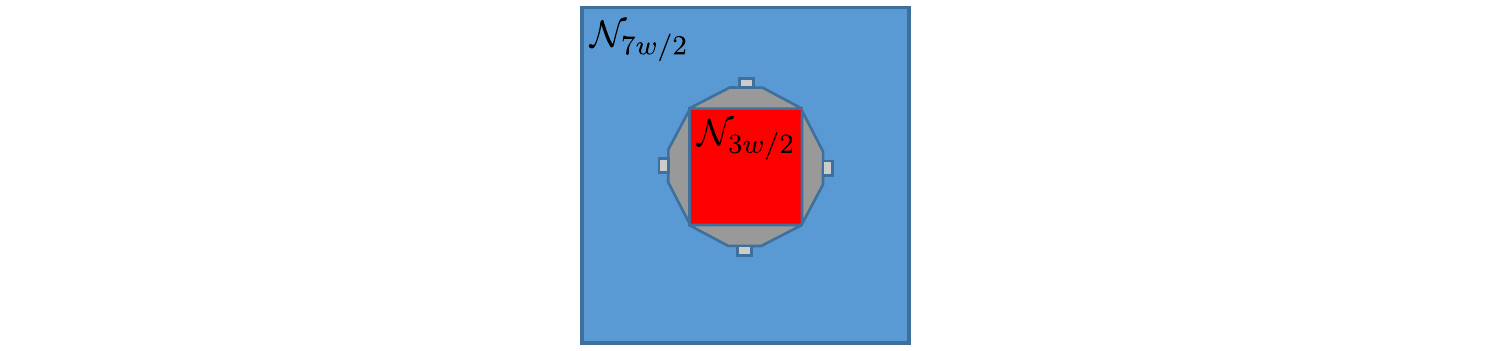}
\caption{Neighborhoods described in the proof of Lemma \ref{Lemma:Monoch_spread2}.}
\label{fig:Lemma_monoch_spread2}
\end{figure}

\begin{lemma*} \label{Lemma:Monoch_spread2}
Consider a good block at the center of $\mathcal{N}_{\radius}$ with ${\radius} > m$. A $3w/2$-block with (+1) agents at the center of a $7w/2$-block contained in the good block will make all (-1) agents right outside the $3w/2$-block unhappy with probability one and with at most $(3w/4+1)^2$ flips happening before $T({\radius})$,  for sufficiently large ${\size}$.
\end{lemma*}

\begin{proof}
Consider four identical isosceles trapezoids outside the $3w/2$-block whose larger bases are the sides of the $3w/2$-block (see Figure \ref{fig:Lemma_monoch_spread2}). Let $\zeta = (3-8\tau)/2$ and $\nu = (16\tau -5)/6$. Let the smaller bases of the above trapezoids be $2(3/4-2\zeta)w$ and their heights be $2 \nu w$. For $\tau > 0.3463$, since these trapezoids are located inside a good block for sufficiently large $\size$ all the agents of type (-1) in these trapezoids will be unhappy with probability one. Consider the case where these trapezoids have become monochromatic after the flips of (-1) agents happening before $T(\radius)$. Now consider four identical rectangles located outside the trapezoids. Let one side of each of these rectangles be at the center of one of the smaller bases of each of the four trapezoids and of length $2(1/8-\nu)w$ and let the other sides of the triangles be $w/4$. For $\tau > \tau_1$, all the agents of type (-1) located inside these rectangles will be unhappy. Now, as a worst case scenario, let us consider an agent outside the  $3w/2$-block and next to its corner which shares the smallest number of agents with the monochromatic regions. When the unhappy agents in the rectangles flip before $T(\radius)$, for this agent to be unhappy we need to have
\begin{align*}
\left[1-\frac{1}{4}- \left(\frac{1}{4}+\frac{1}{2}-\zeta \right)\nu - \frac{1}{4}\left(\frac{1}{8}-\nu \right)\right]\frac{1}{2} + \frac{o({\size})}{{\size}}<\tau,
\end{align*}
which can be simplified to (\ref{Eq:tau2_eq}).
This means that for $\tau < \tau_1$ and for sufficiently large $\size$ this agent will be unhappy with probability one. Since all the other agents of type (-1) right outside the $3w/2$-block share at least the same number of agents with the single-type regions, we have that for sufficiently large $\size$, all the (-1) agents right outside the $3w/2$-block will be unhappy with probability one.  
 \end{proof}


\begin{figure}[!t]
\centering
\includegraphics[width=\textwidth]{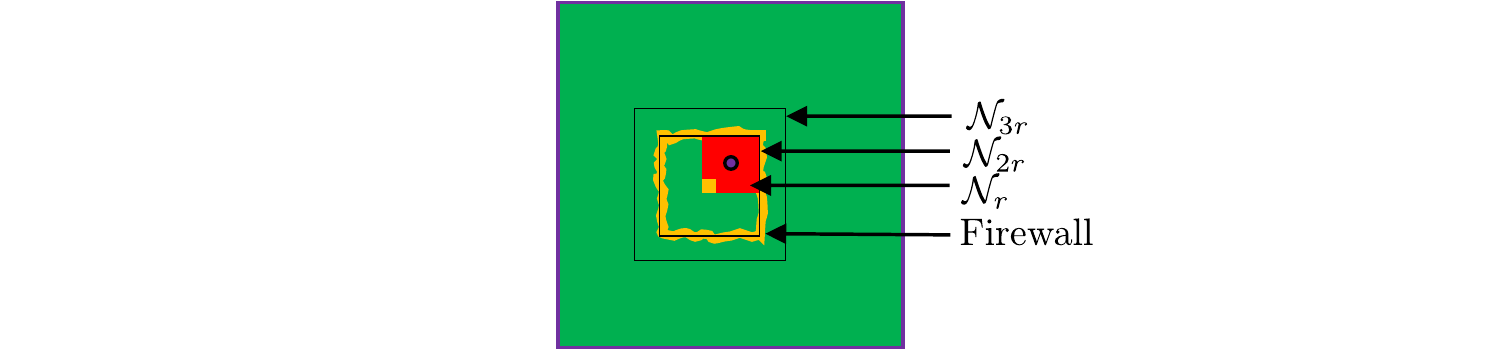}
\caption{Neighborhoods described in the proof of Lemma~\ref{Lemma:fw2}. Agent $u$ is depicted by the circle in the red square and the $\epsilon'$-radical region is depicted by the small orange square in the red square.}
\label{fig:thrm_fw2}
\end{figure}
\

The following lemma, which can be thought of as the counterpart of Lemma \ref{Lemma:fw} for $\tau \in (\tau_2,\tau_1]$, shows that conditional on some events, the size of the almost monochromatic region of an arbitrary agent is exponential in $\size$. Unless otherwise stated, by a good block we mean a good $\bsize$-block and by a bad block we mean a bad $\bsize$-block.

\begin{lemma*}\label{Lemma:fw2}
 Let $\mathcal{N}_{\radius}$, $\mathcal{N}_{{\radius}/2}$, $\mathcal{N}_{4r}$, and $\mathcal{N}_{r}$ be  all centered at  $u$  with  
 \begin{align*}
 {\radius} = 2^{[1-H(\tau')]{\size}/2}, \\
 r = 2^{[1-H(\tau')]{\size}/2-o({\size})}, 
 \end{align*}
 and $r<{\radius}/8$. Let $u^+$ denote an arbitrary (+1) agent, $T({\radius})$ be as defined in (\ref{Tinfdef}), and $\kappa>0$ be such that $\kappa  r {\size}^{3/2}$ is the total number of agents in a $2r$-chemical path. Conditioned on the following events, w.h.p. the almost monochromatic region of $u$ will have at least radius $r$.
\begin{enumerate} 
\item  $A = \left\{\forall v\in \mathcal{N}_{\radius}, \;  u^{+}   \mbox{ { would be happy at the location of $v$ at $t=0$}}\right\}, $
\item  $\textit{B = \{$T({\radius}/2) > 2\kappa  r {\size}^{3/2}$\}}$,
\item  $C = \text{\{$\mathcal{N}_r$  { contains a $2r$-expandable radical region at $t=0$\}}},$
\item  $D = \mbox{\{{$\not \exists$ cluster of bad blocks with $l_1$-radius} $r' > {\size}^2$ { in} $\mathcal{N}_{4r}$ at $t=0$\}},$
\item  $E = \text{\{${\size}_B/{\size}_G < e^{-{\size}^\epsilon}$ in $\mathcal{N}_{r}$ at $t=0$\}}$, where ${\size}_B$ is the number of bad blocks and ${\size}_G$ is the number of good blocks in $\mathcal{N}_{r}$.
\end{enumerate}
\end{lemma*}

\begin{proof}
Conditional on events $A, B$, and $C$, w.h.p. a $2r$-expandable radical region will lead to the formation of a firewall that contains $\mathcal{N}_r$. With additional conditioning on events $D$ and $E$ once the firewall is formed, the expandable radical region will turn all the interior of at least $\mathcal{N}_r$ almost monochromatic by the end of the process. 
Let $M(r)$ denote the event that the radius of the almost monochromatic region of $u$ is at least $r$. 
Let $T_f$ be the time at which the firewall forms, i.e., its agents become monochromatic. We have
\begin{align*}
P\left(M(r)\given[\Big]A, B , C , D, E \right) \ge P\left(T_f<2\kappa  r \sqrt{{\size}}  \given[\Big]A,B,C,D,E \right).
\end{align*}
Let $T'_f$ be the sum of $\kappa  r {\size}^{3/2}$ exponential random variables with mean one, where $\kappa  r {\size}^{3/2}$ is the total number of agents in the $2r$-chemical path. It is easy to see that $T'_f$ is an upper bound for the time it takes until the firewall is formed, i.e., all agents inside the firewall flip to (+1), one by one. Hence, we have
\begin{align*}
  P\left(M(r)\given[\Big]A,B,C,D,E\right) \ge P\left(T'_f<2\kappa  r \sqrt{{\size}}\right). 
\end{align*}
Next we bound this probability. We have
\begin{align*}
P\left(T'_f\ge 2\kappa  r {\size}^{3/2}\right) &\le P\left(|T'_f-\mathbb{E}[T'_f]|\ge \kappa r {\size}^{3/2}\right).
\end{align*}
By Chebyshev's inequality we have
\begin{align*}
P\left(T'_f\ge 2\kappa  r {\size}^{3/2}\right) = O\left(\frac{VarT'_f}{(r{\size}^{3/2})^2}\right) =O\left(\frac{r\sqrt{{\size}}}{(r{\size}^{3/2})^2}\right) = O \left(\frac{1}{r{\size}^{3/2}}\right),
\end{align*}
leading to the desired result.  
\end{proof}

With the above definitions and results, we can proceed to the first part of the proof of Theorem~\ref{Thrm:second_theorem} (for $\tau_2< \tau \leq \tau_1$).

\

\noindent {\bf Proof of Theorem~\ref{Thrm:second_theorem} {\normalfont (for $\tau_2< \tau \le \tau_1$)}:}
First, we derive the lower bound in the theorem letting
\begin{equation}
a(\tau) = \left[1-(2\epsilon'+\epsilon'^2)\right]\left[1-H(\tau')\right],
\end{equation}
where $\epsilon' > f(\tau)$, and $\tau'=(\tau {\size} -2)/({\size}-1)$.
 

We consider  neighborhoods $\mathcal{N}_{\radius}$,  $\mathcal{N}_{\radius/2}$,  and  $\mathcal{N}_r$,  with ${\radius}  =  2^{[1-H(\tau')]{\size}/2}$ and $r< \radius/8$,
all centered at node $u$   as depicted in Figure~\ref{fig:thrm2}. 
We let
${\radius}  =  2^{[1-H(\tau')]{\size}/2}$,
and $u^{+}$ be an arbitrary (+1) agent and consider the following event in the initial configuration
\begin{align*}
A = \{\forall v\in \mathcal{N}_{\radius},   u^{+}   \mbox{ would be happy at the location of } v \mbox{ at } t=0 \}. 
\end{align*}
By Lemma \ref{Lemma:R_unhappy} of the Appendix, we have
\begin{align} 
P(A) \rightarrow 1,\  \mbox{ as } \    {\size} \rightarrow \infty.
\label{pa2}
\end{align}
We then consider a chemical firewall of radius $2r$ centered anywhere inside $\mathcal{N}_r$, let $\kappa >0$   so that $\kappa  r {\size}^{3/2}$ is an upper bound on the total number of agents in the $2r$-chemical path containing it, and
consider the event  
\begin{align*}
\textit{B = \{$T({\radius}/2) > 2\kappa  r {\size}^{3/2}$\}},
\end{align*}
where $T({\radius})$ is defined in (\ref{Tinfdef}).
By  Lemma~\ref{Lemma:Unhappy_growth}, we can choose $r$ proportional to ${\radius}/({\size}^3)$ so  that
\begin{align}
P\left(B\given A\right) \rightarrow 1,\  \mbox{ as } \    {\size} \rightarrow \infty.
\label{pea2}
\end{align}
With this choice, we also have
\begin{align*}
r &= 2^{[1-H(\tau')]{\size}/2-o({\size})},
\end{align*}
and if we consider the event
\begin{align*}
C = \text{\{$\mathcal{N}_r$  contains a $2r$-expandable radical region at $t=0$\}},
\end{align*}
by Lemma \ref{Lemma:expandable} and Lemma \ref{Lemma:firewall_path} and the FKG inequality,  since $\epsilon'>f(\tau)$ we conclude  that for sufficiently large ${\size}$ 
\begin{align}
P(C) \ge 2^{-[1-H(\tau')][2\epsilon'+(\epsilon')^2]{\size} -o({\size})}, 
\label{pc2}
\end{align}
and there is a $2r$-expandable radical region surrounding $u$.
Let us divide the grid into $m$-blocks in the obvious way. Let the \textit{radius} of a bad cluster be defined as 
\begin{align*}
\sup \{\Delta(0,x): x \in \mbox{bad cluster}\}.
\end{align*}
where $\Delta$ denotes the $l_1$ distance.
Let
\begin{align*}
D = \mbox{\{{$\not \exists$ cluster of bad blocks with $l_1$-radius} $r' > {\size}^2$ { blocks in} $\mathcal{N}_{4r}$ at $t=0$\}}.
\end{align*} 
By Lemma \ref{Lemma:bad_cluster}, we have
\begin{align} 
P(D) \rightarrow 1,\  \mbox{ as } \    {\size} \rightarrow \infty.
\label{pd2}
\end{align}
Finally, let $\epsilon \in (0,1/2)$ and let ${\size}_B$ and ${\size}_G$ denote the total number of bad blocks sharing at least one agent with $\mathcal{N}_r$ and good blocks contained in $\mathcal{N}_r$ respectively and let
\begin{align*}
E = \text{\{${\size}_B/{\size}_G < e^{-{\size}^\epsilon}$ in $\mathcal{N}_{r}$ at $t=0$\}}.
\end{align*}
By an application of Lemma \ref{Lemma:bad_ratio}, also
\begin{align} 
P(E) \rightarrow 1,\  \mbox{ as } \    {\size} \rightarrow \infty.
\label{pe2}
\end{align}
See Figure \ref{fig:thrm2} for a visualization of the neighborhoods defined above.

 \begin{figure}[!t]
\centering
\includegraphics[width=\textwidth]{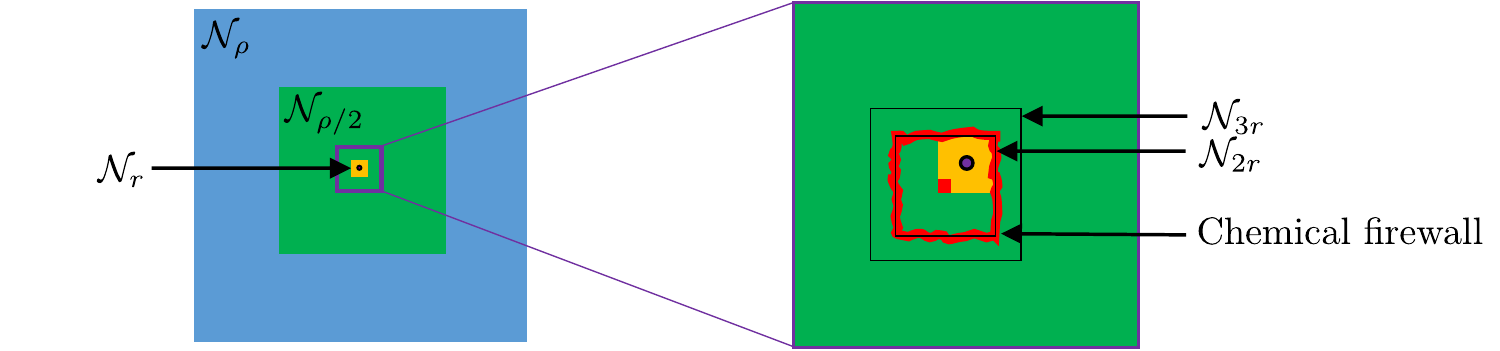}
\caption{Neighborhoods described in the proof of Theorem~\ref{Thrm:second_theorem}.}
\label{fig:thrm2} 
\end{figure}
 
Now it is easy to see that the events $A$, $B$, $C$, $D$, and $E$ are increasing. By combining (\ref{pa2}), (\ref{pea2}), (\ref{pc2}), (\ref{pd2}), and (\ref{pe2}), and using  a version of the Fortuin-Kasteleyn-Ginibre (FKG) inequality adapted  to our dynamic process described in Lemma~\ref{Lemma:FKG_Harris} of the Appendix, it follows  that for ${\size}$ sufficiently large 
\begin{align}
P(A \cap B \cap C \cap D \cap E) &\ge P(A)P(B)P(C)P(D)P(E) \\
&\ge P(A)P(A\cap B)P(C)P(D)P(E) \nonumber \\
&=  P(B|A)[P(A)]^2P(C)P(D)P(E) \nonumber \\
&=  2^{-[1-H(\tau')](2\epsilon'+(\epsilon')^2){\size} -o({\size})}. \label{Eq:intersect2}
\end{align}
 
 Since by Lemma~\ref{Lemma:fw2} we have that conditional on  $A, B, C ,D, E$, at the end of the process w.h.p. agent $u$ will be  part of an almost monochromatic region with radius at least $r$, it follows that (\ref{Eq:intersect2}) is also a lower bound for the probability that the monochromatic neighborhood of agent $u$ will have   size of at least proportional to $r^2$. The desired lower bound on the expected size of the monochromatic region now easily follows by multiplying  (\ref{Eq:intersect2}) by the size  of a neighborhood of radius $r$.
 The second part of the proof follows the same argument as the second part of the proof of Theorem~\ref{Thrm:first_theorem}.  
\subsection{Extension to the interval $1/2<\tau<1-\tau_2$}
We call  \textit{super-unhappy agents} the unhappy agents that can potentially become happy once they flip their type.
While for $\tau<1/2$ unhappy agents  can alway  become happy by flipping their type,  for $\tau>1/2$ this is only true for the super-unhappy agents. 
It follows that for  $\tau > 1/2$ super-unhappy agents act  in the same way as unhappy agents do for   $\tau < 1/2$. 


We let $\bar{\tau} = 1 - \tau + 2/{\size}$. A \textit{super-unhappy agent} of type (-1) is an agent for which $W < \bar{\tau}{\size}$ where $W$ is the number of (-1) agents in its neighborhood. The reason for adding the term $2/{\size}$ in the definition is to account for the strict inequality that is needed for being unhappy and the flip of the agent at the center of the neighborhood which adds one agent of its type to the neighborhood. 
A \textit{super-radical region} is a neighborhood $\mathcal{N}_S$ of radius $S=(1+\epsilon')w$ such that $W_S <  \bar{\tau}'(1+\epsilon')^2{\size} $, where    $\epsilon \in (0,1/2)$ and
\begin{equation*}
\bar{\tau}' = \left(1-\frac{1}{\bar{\tau} {\size}^{1/2-\epsilon}}\right)\bar{\tau}. 
\end{equation*}

By replacing $\tau$ with $\bar{\tau}$, ``unhappy agent" with ``super-unhappy agent"  and   ``radical region" with  ``super-radical region,"   it can be checked that all proofs extend to the interval $1/2<\tau<1-\tau_2$.


\section{Concluding Remarks }  \label{Sec:Conclusion}

The main  lesson learned  from our study is that even a small amount of intolerance can lead to segregation at the large scale. We remark, however, that the model is somewhat naturally biased towards segregation because agents can flip their type when a sufficiently large number of their neighbors are different from themselves, but they never flip when a  large number of their neighbors are of their same type. Variations   where agents could potentially flip in both situations, namely they are ``uncomfortable" being  both a minority or a majority in a largely segregated area, would be of interest. Another direction of further study could be the investigation of how the parameter of the initial  distribution  of the agents influences segregation,  since it is only known that complete segregation occurs w.h.p.\  for $\tau=1/2$ and $p \in (1-\epsilon,1)$, while we have shown that for $0.344 < \tau < 1/2$ and $p=1/2$ the size of the monochromatic region  is at most exponential in the size of its neighborhood, w.h.p. We also point out that for $\tau = 1/2$ and for   $\tau \in [1/4,\tau_2] \cup [1-\tau_2,3/4]$ the behavior of the model is unknown. Finally,    our results only show lower bounds on the expected size of the monochromatic region containing a given agent, but they do not show that in the steady state every agent ends up  in an exponentially large monochromatic region with high probability. A possibility that is consistent with these results (but inconsistent with the simulation results) is that only an exponentially small fraction of the nodes are contained in large monochromatic regions at the end of the  process, but that those regions are so large that the expected radius of the monochromatic region containing any node is exponentially large. Proving an exponential lower bound on the size of the monochromatic region w.h.p., rather than in expectation, would rule out this possibility.


\section*{Acknowledgment}
The authors  thank Prof. Jason Schweinsberg  of the Mathematics Department of  University of California at San Diego for providing invaluable feedback on earlier drafts of the paper and for suggesting some improved proofs.

\printbibliography

\newpage
\section{Appendix}








\subsection{Concentration bound on the number of agents in the initial configuration} \label{Sec:A_Sufficient}

\begin{lemma*}\label{Lemma:balanced}
Let $\epsilon \in (0,1/2)$, and let $\mathcal{N}$ be an arbitrary neighborhood in the grid with ${\size}$ agents. 
There exist $c,c' \in \mathbb{R}^+$, such that 
\begin{align}\label{Eq:balanced}
P\left(|W - {\size}/2| < c{\size}^{1/2+\epsilon}\right) \ge 1-2e^{-c'{\size}^{2\epsilon}}.
\end{align}
\end{lemma*} 

\begin{proof} Let $W_i$ be the random variable associated with the type of the $i$'th agent in $\mathcal{N}$ such that it is one whenever the type is (-1) and zero otherwise. Let $\mathcal{F}_i = \sigma(W_1,...,W_i)$. Then it is easy to see that $M_n = \mathbb{E}[W|\mathcal{F}_n]$ for $n=1,...,{\size}$ is a martingale. It is also easy to see that $M_0 = \mathbb{E}[W] =  {\size}/2$, and $M_{{\size}} = W$. We also have 
\begin{align*}
|M_n - M_{n-1}|  &= \left|\mathbb{E}\left(\sum_{i=1}^{{\size}}W_i|\mathcal{F}_n\right) - \mathbb{E}\left(\sum_{i=1}^{{\size}}W_i|\mathcal{F}_{n-1}\right)\right|  \\
&= \left|W_n+ ({\size}-n)/2 - [{\size}-(n-1)]/2\right| \\
&\le \left|W_n - \frac{1}{2}\right| \le 1/2,
\end{align*}
for  $n=1,2,...,{\size}$. Now using Azuma's inequality, there exist constants $c_1,c_2 \in \mathbb{R}^+$ such that  
\begin{align*}
P\left(W - {\size}/2 \ge c{\size}^{1/2+\epsilon}\right) \le e^{-c_1{\size}^{2\epsilon}},
\end{align*} 
and
\begin{align*}
P\left(W - {\size}/2 \le -c'{\size}^{1/2+\epsilon}\right) \le e^{-c_2{\size}^{2\epsilon}}.
\end{align*} 
It follows by an application of Boole's inequality  that there exists a constant $c\in \mathbb{R}^+$ such that (\ref{Eq:balanced}) holds.   
\end{proof}

\subsection{Preliminary results for the proof of Theorem~\ref{Thrm:first_theorem}} \label{Sec:A_proof_theorem_1}

First, we give a bound on the probability of having an unhappy agent in the initial configuration, we then extend this bound for a radical region.
\begin{lemma*} \label{Lemma:unhappyprob}
Let $p_u$ be the probability of being unhappy for an arbitrary agent in the initial configuration. There exist positive constants $c_l$ and $c_u$ which depend only on $\tau$ such that
\begin{align*}
c_{l}\frac{2^{-[1-H(\tau')]{\size}}}{\sqrt{{\size}}}  \le p_u \le c_{u}\frac{2^{-[1-H(\tau')]{\size}}}{\sqrt{{\size}}}.
\end{align*}
 where $\tau' = \frac{\tau {\size} - 2}{{\size}-1}$, and $H$ is the binary entropy function.
\end{lemma*}

\begin{proof}
We have 
\begin{align} \label{eq:pu}
p_u = \frac{1}{2^{\size}} \sum_{k = 0 }^{\tau {\size} - 2}{{{\size}-1}\choose{k}} +   \frac{1}{2^{\size}} \sum_{k = 0 }^{\tau {\size} - 2}{{{\size}-1}\choose{k}},
\end{align}
where the two unit reduction is to account for the strict inequality and the agent at the center of the neighborhood. Let $\tau' = \frac{\tau {\size} - 2}{{\size}-1}$. After some   algebra we have
\begin{align*}
{{\size}-1 \choose \tau' ({\size}-1)} \le \sum_{k = 0 }^{\tau' ({\size}-1)}{{{\size}-1}\choose{k}} \le \frac{1-\tau'}{1-2\tau'}{{\size}-1 \choose \tau' ({\size}-1)},
\end{align*}
and using Stirling's formula, there exist constants $c,c' \in \mathbb{R}^+$ such that
\begin{align*}
{c\frac{2^{-[1-H(\tau')]({\size}-1)}}{\sqrt{({\size}-1)\tau'(1-\tau')}}} \le {{\size}-1 \choose \tau' ({\size}-1)} \le c'\frac{2^{-[1-H(\tau')]({\size}-1)}}{\sqrt{({\size}-1)\tau'(1-\tau')}}.
\end{align*}
The result follows by combining the above inequalities.  
\end{proof}

\begin{lemma*} \label{Lemma:super_unhappy_prob}
There exist positive constants $c_l$ and $c_u$ which depend only on $\tau$ such that in the initial configuration, an arbitrary neighborhood with radius $(1+\epsilon')w$ is a radical region with probability $p_{\epsilon'}$ where we have
\begin{align*}
c_{l}{2^{-[1-H(\tau'')](1+\epsilon')^2{\size}-o({\size})} \le p_{\epsilon'} \le c_{u} 2^{-[1-H(\tau'')](1+\epsilon')^2{\size}+o({\size})}},
\end{align*}
where $\tau'' = (\lfloor \hat{\tau}(1+\epsilon')^2{\size} \rfloor - 1) / (1+\epsilon')^2{\size} $, $\hat{\tau} = (1-{1}/{(\tau {\size}^{1/2-\epsilon})})\tau$, and $H$ is the binary entropy function.
\end{lemma*}
\begin{proof}
The proof follows the same lines as in the proof of Lemma \ref{Lemma:unhappyprob}.  
\end{proof}


\begin{lemma*} \label{Lemma:R_unhappy}
Let ${\radius}  =  2^{[1-H(\tau')]{\size}/2}$ and 
\begin{align*}
A = \left\{\forall v\in \mathcal{N}_{\radius}, \;  u^{+}   \mbox{ {\normalfont would be happy at the location of $v$ at  $t=0$}}\right\}. 
\end{align*}
 Then $A$ occurs w.h.p. 
\end{lemma*}

\begin{proof}
Let $U_i$ for $i=1,2,...,|\mathcal{N}_{\radius}|$ be the event that agent $u^+$ would be happy at the location of $i$'th agent of $\mathcal{N}_{\radius}$. It is easy to see that $P(U_i)=p_u$ (see~(\ref{eq:pu})). Hence we have
\begin{align*}
P(A) &= P\left(U_1^C \cap ... \cap U^C_{|\mathcal{N}_{\radius}|}\right) \\
&= 1 - P\left(U_1 \cup ... \cup U_{|\mathcal{N}_{\radius}|}\right)  \\
&\ge  1 - |\mathcal{N}_{\radius}|\frac{2^{-[1-H(\tau')]{\size}}}{\sqrt{{\size}}} \\
&\ge 1 - \frac{5}{\sqrt{{\size}}} 
\end{align*}
which tends to one as ${\size} \rightarrow \infty$.  
\end{proof}

The following lemma gives a simple lower bound for the probability of having a radical region inside a neighborhood which has radius $r= 2^{[1-H(\tau')]{\size}/2-o({\size})}$. We call a radical region with radius $(1+\epsilon')w$ an $\epsilon'$-radical region.

\begin{lemma*} \label{Lemma:r_unhappy}
Any arbitrary neighborhood $\mathcal{N}_r$ with radius $r= 2^{[1-H(\tau')]{\size}/2-o({\size})}$ in the initial configuration has at least one $\epsilon'$-radical region in it with probability at least $2^{-[1-H(\tau')](2\epsilon'+\epsilon'^2){\size} -o({\size})}$.
\end{lemma*}

\begin{proof}
Divide the neighborhood into $2(1+\epsilon')w$-blocks, and let ${\size}_b$ denote the number of blocks in $\mathcal{N}_r$. Define the events 
\begin{align*}
Q_i = \{ \text{The i-th block of } \mathcal{N}_r \text{ is an } \epsilon'\text{-radical region} \}, \\
Q = \{\text{There is an } \epsilon'\text{-radical region in }\mathcal{N}_r \}.
\end{align*}
Using Lemma \ref{Lemma:super_unhappy_prob}, it follows that
\begin{align*}
P(Q) \ge& \  P\left(Q_1 \cup ... \cup Q_{{\size}_b}\right) \\
=& \  1 - P\left(Q_1^C \cap ... \cap Q^C_{{\size}_b}\right) \\
=&  \ \frac{4r^2}{(1+\epsilon')^2{\size}}2^{-[1-H(\tau'')](1+\epsilon')^2{\size}-o({\size})}\\
=& \ 2^{-[1-H(\tau')][2\epsilon'+\epsilon'^2]{\size} -[H(\tau')-H(\tau'')](1+\epsilon')^2{\size}-o({\size})} \\
=& \  2^{-[1-H(\tau')][2\epsilon'+\epsilon'^2]{\size} -o({\size})}.
\end{align*}  
\end{proof}

\subsection{FKG-Harris inequality} \label{FKGus}
The following is Theorem~4 in \cite{liggett2010stochastic} which is originally by Harris \cite{harris1977correlation}. Let $\sigma_t$ be the configuration of the agents on the grid at time $t$. Let $\mathbb{E}^{\sigma_0}[X]$ be the expected value of the random variable $X$, when the initial state of the system is $\sigma_0$. A probability distribution $\mu$ on $\{0,1\}^{\mathbb{Z}^d}$ is said to be positively associated if for all increasing $f$ and $g$ we have
\begin{align*}
\mathbb{E}[f(\sigma)g(\sigma)] \ge \mathbb{E}[f(\sigma)]\mathbb{E}[g(\sigma)].
\end{align*}
\begin{theorem*}[Harris] \label{Thrm:harris}
Assume the process satisfies the following two properties:
(a) Individual transitions affect the state at only one site.
(b) For every continuous increasing function $f$ and every $t>0$, the function $\sigma_0 \rightarrow \mathbb{E}^{\sigma_0}[f(\sigma_t)]$ is increasing. Then, if the initial distribution is positively associated, so is the distribution at all later times.
\end{theorem*}

The following is a version of the FKG inequality~\cite{fortuin1971correlation} in our setting. The original inequality holds for a static setting and  is extended here to our time-dynamic setting using Theorem~\ref{Thrm:harris}.
\begin{lemma*}[FKG-Harris] \label{Lemma:FKG_Harris}
Let $A$ and $B$ be two increasing events defined on our process on the grid. We have
\begin{align*}
P(A\cap B) \ge P(A)P(B).
\end{align*}
\end{lemma*}
\begin{proof}
Assume $A$ and $B$ are increasing random variables which depend only on the states of the sites $v_1,v_2,...,v_k$ and first time step. We proceed by induction on $k$. First, let $k=1$. Let $\omega(v_1)$ be the realization of the site $v_1$. We also have
\begin{align*}
\left(1_A(\omega_1)-1_A(\omega_2)\right) \left(1_B(\omega_1)-1_B(\omega_2) \right) \ge 0,
\end{align*}
for all pairs of vectors $\omega_1$ and $\omega_2$ from the sample space.
We have
\begin{align*}
0 &\le \sum_{\omega_1,\omega_2} \left(1_A(\omega_1)-1_A(\omega_2)\right) \left(1_B(\omega_1)-1_B(\omega_2)\right)P(\omega(v_1)=\omega_1)P(\omega(v_1)=\omega_2) \\
&=2\left(P(A\cap B)-P(A)P(B)\right),
\end{align*}
as required. Assume now that the result is valid for values of $n$ satisfying $k<n$. Then
\begin{align*}
P(A\cap B) &= \mathbb{E}\left[P\left(A\cap B \given[\Big]\omega(v_1),...,\omega(v_{n-1})\right) \right] \\
&\ge \mathbb{E}\left[P\left(A\given[\Big]\omega(v_1),...,\omega(v_{n-1})\right)P\left(B\given[\Big]\omega(v_1),...,\omega(v_{n-1})\right) \right],
\end{align*}
since, given $\omega(v_1),...,\omega(v_{n-1})$, $1_A$ and $1_B$ are increasing in the single variable $\omega(v_n)$. Now since $P\left(A|\omega(v_1),...,\omega(v_{n-1})\right)$ and $P\left(B|\omega(v_1),...,\omega(v_{n-1})\right)$ are increasing in the space of the $n-1$ sites, it follows from the induction hypothesis that 
\begin{align} \label{Eq:fkg_1}
P(A\cap B) &\ge \mathbb{E}\left[P\left(A\given[\Big]\omega(v_1),...,\omega(v_{n-1})\right)\right]\mathbb{E}\left[P\left(B\given[\Big]\omega(v_1),...,\omega(v_{n-1})\right) \right] \nonumber \\
&= P(A)P(B). 
\end{align}

Next, assume $A$ and $B$ are increasing random variables which depend only on the states of the sites in the first $k$ time steps. We proceed by induction on $k<K$ such that $K$ denotes the final time step over all the realizations. First, let $k=0$. Let $\omega(t_0)$ be the configuration of the graph at the first time step. We have
\begin{align*}
P(A\cap B) \ge P(A)P(B),
\end{align*}
by the above result. Assume now that the result is valid for all values of $k$ satisfying $k<K$. Then, since our process satisfies the conditions of   Theorem~\ref{Thrm:harris}  
and given $\omega(t_0),...,\omega(t_{K-1})$, $1_A$ and $1_B$ 
are increasing in $\omega(t_{K})$, we have
\begin{align*}
P(A\cap B) &= \mathbb{E}\left[P\left(A\cap B \given[\Big]\omega(t_0),...,\omega(t_{K-1})\right) \right] \\
&\ge \mathbb{E}\left[P\left(A\given[\Big]\omega(t_0),...,\omega(t_{K-1})\right)P\left(B\given[\Big]\omega(t_0),...,\omega(t_{K-1})\right) \right].
\end{align*}
 Now, since 
$P\left(A|\omega(t_0),...,\omega(t_{K-1})\right)$ and  
$P\left(B|\omega(t_0),...,\omega(t_{K-1})\right)$ are increasing in the space of the configurations of the graph in the first $K-1$ time steps, it follows from the induction hypothesis that 
\begin{align*}
P(A\cap B) &\ge \mathbb{E}\left[P\left(A\given[\Big]\omega(t_0),...,\omega(t_{K-1})\right)\right]\mathbb{E}\left[P\left(B\given[\Big]\omega(t_0),...,\omega(t_{K-1})\right) \right]\\
&= P(A)P(B). 
\end{align*}  
\end{proof}


%
%

\end{document}